%% file: RR-8394.tex
\documentclass[twoside]{article} 

\usepackage[a4paper]{geometry}
\usepackage[utf8]{inputenc}
\usepackage[T1]{fontenc}
\usepackage{RR}
\usepackage{hyperref}
\usepackage{amsmath}
\usepackage{amssymb}
\usepackage{amsthm}
\usepackage{dcolumn}

\usepackage{mdwlist}
\usepackage{algorithmic}
\usepackage{algorithm}
\usepackage{enumerate}
\usepackage{comment}
\usepackage{url}
\usepackage{alltt}
\usepackage[outerbars,xcolor]{changebar}
\usepackage{backref}

\usepackage[colorinlistoftodos]{todonotes} 
\newcommand{\je}[2][]{\todo[color=pink, #1]{JE: {\small #2}}}

\usepackage{pgf}
\usepackage{tikz}
\usepackage{pgflibraryshapes}

\newenvironment{smalltt}%
{\footnotesize\begin{alltt}{}}{\end{alltt}}
\newcommand{\smtt}[1]{{\footnotesize\texttt{#1}}}
\newcommand{\uriref}[2]{{\footnotesize
   \texttt{#2}}}
\newcommand{\blank}[1]{{\tt \small {?#1}}}
\newcommand{\litteral}[1] {\texttt{\small "#1"}}

\newcommand{\mparagraph}[1] {\noindent {\bf #1.}}

\newcommand{\furiref}[2]{#2}

\newdimen\argwidth 
\def\[[#1\]]{%
\setbox0=\hbox{$#1$}\argwidth=\wd0 
\setbox0=\hbox{$\left[\box0\right]$}\advance\argwidth by -\wd0 
\left[\kern.3\argwidth\box0\kern.3\argwidth\right]}

\newtheorem{theorem}{Theorem}
\newtheorem{corollary}{Corollary}

\newtheorem{lemma}{Lemma}
\newtheorem{definition}{Definition} 
 
\newtheorem{proposition}{Proposition} 
\newtheorem{example}{Example}

\RRNo{8394}
\RRdate{November 2013}
\RRauthor{
Faisal Alkhateeb\thanks{Computer Science Department, Yarmouk University, Jordan}
\and J\'er\^ome Euzenat\thanks{INRIA \& LIG, Grenoble, France}
}
\authorhead{Faisal Alkhateeb \and J\'er\^ome Euzenat}

\RRtitle{Évaluation de requêtes SPARQL en fonction d'un schéma RDF par les chemins}
\RRetitle{Answering SPARQL queries modulo RDF Schema with paths}
\titlehead{Answering SPARQL queries modulo RDF Schema with paths}
\RRnote{This research report sums up our knowledge about the various ways to query RDF graph modulo RDFS, providing improvements on all the existing proposals. Part of it will be published as \cite{alkhateeb2014a}.}
\RRresume{SPARQL est le langage de requête standard pour interroger des graphes RDF. 
Dans son instanciation stricte, il ne propose que des requêtes en fonction de la sémantique de RDF et n'interprète donc pas les vocabulaires exprimés en RDFS ou OWL.
Plusieurs extensions de SPARQL ont été proposées pour interroger les données RDF en fonction de vocabulaires RDFS et d'ontologies OWL.
Par ailleurs, les extensions de SPARQL qui utilisent des expressions régulières pour naviguer dans les graphes RDF peuvent être utilisées pour répondre aux requêtes sous la sémantique de RDFS.
Nous introduisons un cadre général qui permet d'exprimer d'une manière homogène l'interprétation de SPARQL en fonction de différentes sémantiques.
Les extensions de SPARQL sont interprétées dans ce cadre et leurs inconvénients sont présentés.
En particulier, nous montrons que le langage de requête PSPARQL, une extension stricte de SPARQL, permet de répondre aux requêtes SPARQL sous la sémantique de RDFS avec la même complexité que SPARQL par une transformation des requêtes.
Nous considérons également CPSPARQL, une extension de PSPARQL, qui permet de poser des contraintes sur les chemins. 
Nous montrons que CPSPARQL est suffisamment expressif pour répondre aux requêtes PSPARQL et CPSPARQL sous la sémantique de RDFS. 
Nous présentons également nSPARQL, un langage de chemins inspiré de XPath permettant d'évaluer des requêtes sous la sémantique de RDFS.
Nous comparons l'expressivité et la complexité de nSPARQL et le fragment correspondant de CPSPARQL, que nous appelons cpSPARQL.
Les deux langages ont la même complexité bien que cpSPARQL, étant une extension stricte de SPARQL, soit plus expressif que nSPARQL.
}

\RRabstract{
SPARQL is the standard query language for RDF graphs.
In its strict instantiation, it only offers querying according to the RDF semantics and would thus ignore the semantics of data expressed with respect to (RDF) schemas or (OWL) ontologies.
Several extensions to SPARQL have been proposed to query RDF data modulo RDFS, i.e., interpreting the query with RDFS semantics and/or considering external ontologies.

We introduce a general framework which allows for expressing query answering modulo a particular semantics in an homogeneous way.
In this paper, we discuss extensions of SPARQL that use regular expressions to navigate RDF graphs and may be used to answer queries considering RDFS semantics. 
We also consider their embedding as extensions of SPARQL.
These SPARQL extensions are interpreted within the proposed framework and their drawbacks are presented.
In particular, we show that the PSPARQL query language, a strict extension of SPARQL offering transitive closure, allows for answering SPARQL queries modulo RDFS graphs with the same complexity as SPARQL through a simple transformation of the queries.
We also consider languages which, in addition to paths, provide constraints.
In particular, we present and compare nSPARQL and our proposal CPSPARQL. 
We show that CPSPARQL is expressive enough to answer full SPARQL queries modulo RDFS. 
Finally, we compare the expressiveness and complexity of both nSPARQL and the corresponding fragment of CPSPARQL, that we call cpSPARQL. 
We show that both languages have the same complexity through cpSPARQL, being a proper extension of SPARQL graph patterns, is more expressive than nSPARQL.
}
\RRmotcle{Web Sémantique -- langage de requêtes -- RDF -- RDFS -- SPARQL -- expressions régulières -- expressions régulières avec contraintes -- Chemin - PSPARQL - NSPARQL - nSPARQL - CPSPARQL - cpSPARQL.}

\RRkeyword{Semantic web --  Query language --  Query modulo schema -- Resource Description Framework (RDF) --  RDF Schema --  SPARQL --  Regular expression --  Constrained regular expression -- Path -- PSPARQL -- NSPARQL -- nSPARQL -- CPSPARQL --  cpSPARQL.}
\RRprojets{Exmo}
\URRhoneAlpes 

\begin{document}

\makeRR   

\section{Introduction}

RDF (Resource Description Framework \cite{hayes2004a}) is a standard knowledge representation language dedicated to the description of documents and more generally of resources within the semantic web. 

SPARQL \cite{prudhommeaux2008a} is the standard language for querying RDF data. 
It has been well-designed for that purpose, but very often, RDF data is expressed in the framework of a schema or an ontology in RDF Schema or OWL. 

RDF Schema (or RDFS) \cite{rdfvocabulary} together with OWL \cite{owl} are two ontology languages recommended by the W3C for defining the vocabulary used in RDF graphs. 
Extending SPARQL for dealing with RDFS and OWL semantics when answering queries is thus a major issue. 
Recently, SPARQL 1.1 entailment regimes have been introduced  to incorporate RDFS and OWL semantics \cite{glimm2010b}.
We consider here the case of RDF Schema (RDFS) or rather a large fragment of RDF Schema \cite{munoz2009a}.

Query answering with regard to RDFS semantics can be specified by indicating the inference regime of the SPARQL evaluator, but this does not tell how to implement it.
It is possible to implement a specific query evaluator, embedding an inference engine for a regime or to take advantage of existing evaluators.
For that purpose, one very often transforms the data or the query to standard languages.
Two main approaches may be developed for answering a SPARQL query $Q$ modulo a schema $S$ against an RDF graph $G$: 
the eager approach transforms the data so that the evaluation of the SPARQL query $Q$ against the transformed RDF graph $\tau(G)$ returns the answer, while 
the lazy approach transforms the query so that the transformed query $\tau(Q)$ against the RDF graph $G$ returns the answers. 
These approaches are not exclusive, as shown by \cite{pan2009a}, though no hybrid approach has been developed so far for SPARQL.

There already have been proposals along the second approach. 
For instance, \cite{perez2010a} provides a query language, called nSPARQL, allowing for navigating graphs in the style of XPath. 
Then queries are rewritten so that query evaluation navigates the data graph for taking the RDF Schema into account.
One problem with this approach is that it does not preserve the whole SPARQL: not all SPARQL queries are nSPARQL queries.

Other attempts, such as SPARQ2L \cite{sparq2L} and SPARQLeR \cite{kochut2007a} are not known to address queries with respect to RDF Schema.
SPARQL-DL \cite{sirin2007a} addresses OWL but is restricted with respect to SPARQL.
\cite{kollia2013a} provides a sound and complete algorithm for implementatiing the OWL 2 DL Direct Semantics entailment regime of SPARQL 1.1.
This regime does not consider projection (SELECT) within queries and answering queries is reduced to OWL entailment.

On our side, we have independently developed an extension of SPARQL, called PSPARQL \cite{alkhateeb2009a}, which adds path expressions to SPARQL. 
Answering SPARQL queries modulo RDF Schema could be achieved by transforming them into PSPARQL queries \cite{alkhateeb2008a}. 
PSPARQL fully preserves SPARQL, i.e., any SPARQL query is a valid PSPARQL query. 
The complexity of PSPARQL is the same as that of SPARQL \cite{alkhateeb2008b}. 

Nonetheless, the transformation cannot be generally applied to PSPARQL and thus it is not generally sufficient for answering PSPARQL queries modulo RDFS \cite{alkhateeb2008a}.  
To overcome this limitation, we use an extension of PSPARQL, called CPSPARQL \cite{alkhateeb2007b,alkhateeb2008a}, that uses constrained regular expressions instead of regular expressions.

This report mainly contributes in two different ways, of different nature, to the understanding of SPARQL query answering modulo ontologies.

First, it introduces a framework in which SPARQL query answering modulo a logical theory can be expressed. 
This allows for comparing different approaches that can be used on a semantic basis. 
For that purpose, we reformulate previous work and definitions and show their equivalence. 
This also provides a unified strategy for proving properties that we use in the proof section.




Second, we show that cpSPARQL, a restriction of CPSPARQL, can express all nSPARQL queries with the same complexity. The advantage of using CPSPARQL is that, contrary to nSPARQL, it is a strict extension of SPARQL and cpSPARQL graph patterns are a strict extension of SPARQL graph patterns as well as a strict extension of PSPARQL graph patterns. 
Hence, using a proper extension of SPARQL like CPSPARQL is preferable to a restricted path-based language.
In particular, this allows for implementing the SPARQL RDFS entailment regime.

In order to compare cpSPARQL and nSPARQL, we adopt in this paper a notation similar to nSPARQL, i.e., adding XPath axes, which is slightly different from the original CPSPARQL syntax presented in \cite{alkhateeb2007b,alkhateeb2008a}. After presenting the syntax and semantics of both nSPARQL and CPSPARQL, we show that:
\begin{itemize*}
\item CPSPARQL can answer full SPARQL (and CPSPARQL) queries modulo RDFS (Section~\ref{sec:cpsparqlrdfs});
\item cpSPARQL has the same complexity as nSPARQL and there exist an efficient algorithm for answering cpSPARQL queries (Section~\ref{sec:complexityresult});
\item Any nSPARQL triple pattern can be expressed as a cpSPARQL triple pattern, but not vice versa (Section~\ref{sec:compare}).
\end{itemize*}

\textbf{Outline}. 
The remainder of the report is organized as follows. 
In Section~\ref{sec:preliminaries}, we introduce RDF and the SPARQL language. 
Then we present RDF Schema and the existing attempts at answering queries modulo RDF Schemas (Section~\ref{sec:modulo}).
The PSPARQL language is presented with its main results in Section~\ref{sec:psparql}, and we show how to use them for answering SPARQL queries modulo RDF Schemas. 
Section~\ref{sec:nsparql} is dedicated to the presentation of the nSPARQL query language. 
The cpSPARQL and CPSPARQL languages are presented in detail with their main results in Section~\ref{sec:cpsparql} and we show how to use them for answering SPARQL and CPSPARQL queries modulo RDF Schemas.
Complexity results as well as a comparison between the expressiveness of cpSPARQL and nSPARQL are presented in Section~\ref{sec:comparison}. 
We report on preliminary implementations (Section~\ref{sec:impl}) and discuss more precisely other related work (Section~\ref{sec:relwork}). 
Finally, we conclude in Section~\ref{sec:conclusion}.

\section{Querying RDF with SPARQL}\label{sec:preliminaries}

The Resource Description Framework (RDF \cite{hayes2004a}) is a W3C recommended language for expressing data on the web. 
We introduce below the syntax and the semantics of the language as well as that  of its recommended query language SPARQL.

\subsection{RDF}\label{sec:grdf}\label{sec:rdf} 

This section introduces Simple RDF (RDF with simple semantics, i.e., without considering RDFS semantics of the language \cite{hayes2004a}).

\subsubsection{RDF syntax}\label{sec:grdfsyntax}

RDF graphs are constructed over sets of URI references (or urirefs), blanks, and literals \cite{carroll2004a}. 
Because we want to stress the compatibility of the RDF structure with classical logic, we will use the term variable 
instead of that of ``blank'' which is a vocabulary specific to RDF.
The specificity of blanks with regard to variables is their quantification. Indeed, a blank in RDF is a variable existentially quantified over a particular graph. 
We prefer to retain this classical interpretation which is useful when an RDF graph is placed in a different context. 
To simplify notations, and without loss of generality, we do not distinguish here between simple and typed literals. 

In the following, we treats blank nodes in RDF simply as constants (as if they were URIs) as done in the official specification of SPARQL without considering their existential semantics. 
However, if the existential semantics of blank nodes is considered when querying RDF, the results of this paper indirectly apply by using the graph homomorphism technique \cite{baget2005a}. 

\mparagraph{Terminology} 
An \emph{RDF terminology}, noted $\mathcal T$, is the union of 3 pairwise disjoint infinite sets of \emph{terms}: the set $\mathcal U$ of \emph{urirefs}, the set $\mathcal L$ of \emph{literals} and the set $\mathcal B$ of \emph{variables}. 
The \emph{vocabulary} $\mathcal V$ denotes the set of \emph{names}, i.e., $\mathcal V = \mathcal U \cup \mathcal L$. 
We use the following notations for the elements of these sets: a variable will be prefixed by \blank{} (like \blank{x1}), a literal will be expressed between quotation marks (like \litteral{27}), remaining elements will be urirefs (like \uriref{ex}{price} or simply \smtt{gene}).

\begin{figure}
\begin{center}
\begin{tikzpicture}[scale=0.64]

\draw (-1,5) node[draw,fill=white,rounded corners,x=45pt,y=45pt] (gene) {gene};
\draw (2,4) node[draw,fill=white,rounded corners,x=45pt,y=45pt] (bcd) {bcd};
\draw (8,4) node[draw,fill=white,rounded corners,x=45pt,y=45pt] (cad) {cad};
\draw (0,2) node[draw,fill=white,rounded corners,x=45pt,y=45pt] (tll) {tll};
\draw (5,2) node[draw,fill=white,rounded corners,x=45pt,y=45pt] (hb) {hb};
\draw (10,1.5) node[draw,fill=white,rounded corners,x=45pt,y=45pt] (kni) {kni};
\draw (2,0) node[draw,fill=white,rounded corners,x=45pt,y=45pt] (Kr) {Kr};

\draw[->] (bcd) -- node[above] {\scriptsize{inhibits\_translation}} (cad);
\draw[->] (bcd) -- node[above,sloped] {\scriptsize{inhibits}} (tll);
\draw[->] (bcd) -- node[below,sloped] {\scriptsize{promotes}} (hb);
\draw[->] (bcd) -- node[above,sloped] {\scriptsize{promotes}} (kni);
\draw[->] (bcd) -- node[above,sloped] {\scriptsize{promotes}} (Kr);
\draw[->] (bcd) -- node[above,sloped] {\scriptsize{rdf:type}} (gene);
\draw[->] (tll) -- node[above,sloped] {\scriptsize{regulates}} (Kr);
\draw[->] (tll) -- node[above,sloped] {\scriptsize{rdf:type}} (gene);
\draw[->] (cad) -- node[above,sloped] {\scriptsize{promotes}} (kni);
\draw[->] (hb) -- node[above,sloped,near start] {\scriptsize{inhibits}} (kni);
\draw[->] (hb) -- node[above,sloped] {\scriptsize{promotes}} (Kr);
\draw[->] (kni) -- node[above,sloped] {\scriptsize{inhibits}} (Kr);

\end{tikzpicture}
\end{center}
  
\caption{An RDF graph representing some interactions between genes within the early development of drosophila embryo.}\label{fig:genegraph}
\end{figure}
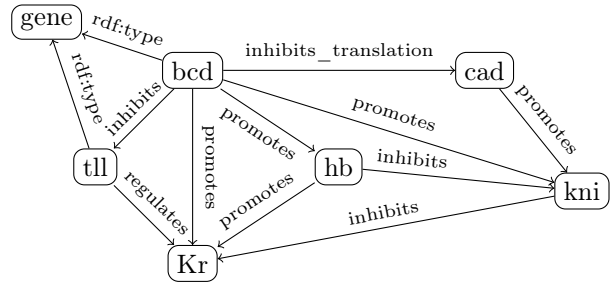



This terminology grounds the definition of RDF graphs and GRDF graphs.
Basically, an RDF graph is a set of triples of the form $\langle subject, predicate, object \rangle$ whose domain is defined in the following definition.

\begin{definition}[RDF graph, GRDF graph] 
An \emph{RDF triple} is an element of $(\mathcal U \cup \mathcal B) \times \mathcal U \times \mathcal T$. An \emph{RDF graph} is a set of RDF triples. 
A \emph{GRDF graph} (for generalized RDF) is a set of triples of $(\mathcal U \cup \mathcal B)  \times (\mathcal U \cup \mathcal B) \times \mathcal T$. 
\end{definition}

So, every RDF graph is a GRDF graph. 
If $G$ is an RDF graph, we use $\mathcal T(G),$ $\mathcal U(G),$ $\mathcal L(G),$ $\mathcal B(G),$ $\mathcal V(G)$ or $voc(G)$ to denote the set of terms, urirefs, literals, variables or names appearing in at least one triple of $G$.

In a triple $\langle s, p, o\rangle$, $s$ is called the subject, $p$ the predicate and $o$ the object. It is possible to associate to a set of triples $G$ a labeled directed multi-graph, 
such that the set of nodes 
is the set of terms appearing as a subject or object at least in a triple of $G$, the set of arcs 
is the set of triples of $G$, 
i.e., if $\langle s, p, o\rangle$ is a triple, then $s\stackrel{p}{\longrightarrow}o$ is an arc (see Figure~\ref{fig:genegraph} and \ref{fig:rdfgraph}).
By drawing these graphs, the nodes resulting from literals are represented by rectangles while the others are represented  by rectangles with rounded corners. In what follows, we do not distinguish the two views of the RDF syntax (as sets of triples or labeled directed graphs). We will then mention interchangeably its nodes, its arcs, or the triples it is made of.

\begin{example}[RDF]\label{ex:rdf}
RDF can be used for exposing a large variety of data.
For instance, Figure~\ref{fig:genegraph} shows the RDF graph representing part of the gene regulation network acting in the fruitfly (\textit{Drosophila melanogater}) embryo.
Nodes represent genes and properties express regulation links, i.e., the fact that the expression of the source gene has an influence on the expression of the target gene. The triples of this graph are the following:

\begin{tabular}{ll}
\begin{minipage}{0.5\linewidth}
\begin{smalltt}
dm:bcd rdf:type rn:gene.
dm:bcd rn:inhibits\_translation dm:cad.
dm:bcd rn:promotes dm:hb.
dm:bcd rn:promotes dm:kni.
dm:bcd rn:promotes dm:Kr.
dm:bcd rn:inhibits dm:tll.
\end{smalltt}
\end{minipage}
&
\begin{minipage}{0.5\linewidth}
\begin{smalltt}
dm:cad rn:promotes dm:kni.
dm:hb rn:inhibits dm:kni.
dm:hb rn:promotes dm:Kr.
dm:kni rn:inhibits dm:Kr.
dm:tll rn:regulates dm:Kr.
dm:tll rdf:type rn:gene.
\end{smalltt}
\end{minipage}
\end{tabular}
This example uses only urirefs.
\end{example}

\begin{example}[RDF Graph]\label{ex:rdfgraph}
RDF can be used for representing information about cities, transportation means between cities, and relationships between the transportation means. The following triples are part of the RDF graph of  Figure~\ref{fig:rdfgraph}: 
\begin{smalltt}
Grenoble TGV Paris .
Paris plane Amman .
TGV subPropertyOf transport .
...
\end{smalltt}
For instance, a triple $\langle$\uriref{}{Paris}, \uriref{}{plane}, \uriref{}{Amman}$\rangle$ means that there exists a transportation mean \uriref{}{plane} from \uriref{}{Paris} to \uriref{}{Amman}.

\begin{figure}
\begin{center}
\begin{tikzpicture}[scale=0.64]
	\begin{scope}
	\tikzstyle{every node}=[draw,fill=white,rounded corners,x=55pt,y=50pt]
		\node (grenoble) at (1.5,0) {\scriptsize{\furiref{ex}{Grenoble}}};
							\node (france) at (0,1) {\scriptsize{\furiref{ex}{France}}};
							\node (paris) at (1.5,1) {\scriptsize{\furiref{ex}{Paris}}};
							\node (madrid) at (3,1) {\scriptsize{\furiref{ex}{Madrid}}};
							\node (spain) at (4.5,1) {\scriptsize{\furiref{ex}{Spain}}};
							\node (italy) at (0,2) {\scriptsize{\furiref{ex}{Italy}}};
							\node (roma) at (1.5,2) {\scriptsize{\furiref{ex}{Roma}}};
							\node (amman) at (3,2) {\scriptsize{\furiref{ex}{Amman}}};
							\node (jordan) at (4.5,2) {\scriptsize{\furiref{ex}{Jordan}}};
							
							\node (tgv) at (0,3.5) {\scriptsize{\furiref{ex}{TGV}}};
							\node (train) at (1.5,3.5) {\scriptsize{\furiref{ex}{train}}};
							\node (plane) at (1.5,4.5) {\scriptsize{\furiref{ex}{plane}}};
							\node (transport) at (3,4) {\scriptsize{\furiref{ex}{transport}}};
								  
	\end{scope}
					
	\draw[->] (grenoble) -- node[left] {\scriptsize{\furiref{ex}{TGV}}}(paris);
	\draw[->] (grenoble) -- node[left] {\scriptsize{\furiref{ex}{cityIn}}}(france);
	\draw[->] (madrid) -- node[right] {\scriptsize{\furiref{ex}{TGV}}}(grenoble);
	\draw[->] (madrid) -- node[above] {\scriptsize{\furiref{ex}{plane}}}(paris);
	\draw[->] (paris) -- node[above] {\scriptsize{\furiref{ex}{cityIn}}}(france);
	\draw[->] (madrid) -- node[above] {\scriptsize{\furiref{ex}{cityIn}}}(spain);
	\draw[->] (amman) -- node[above] {\scriptsize{\furiref{ex}{cityIn}}}(jordan);
	\draw[->] (roma) -- node[above] {\scriptsize{\furiref{ex}{cityIn}}}(italy);
	\draw[->] (paris) -- node[right] {\scriptsize{\furiref{ex}{plane}}}(amman);
	\draw[->] (paris) -- node[left] {\scriptsize{\furiref{ex}{plane}}}(roma);
	\draw[->] (roma) -- node[above] {\scriptsize{\furiref{ex}{plane}}}(amman);
					
	\draw[->] (tgv) -- node[above] {\scriptsize{sp}}(train);
	\draw[->] (plane) -- node[above] {\scriptsize{sp}}(transport);
	\draw[->] (train) -- node[above] {\scriptsize{sp}}(transport);
					
	\draw (9,7) node {{\bf M}};
	\draw (9,0) node {{\bf G}};				
\end{tikzpicture}
\end{center} 
\caption{An RDF graph ($G$) with its schema ($M$) representing information about transportation means between several cities.} \label{fig:rdfgraph}
\end{figure}
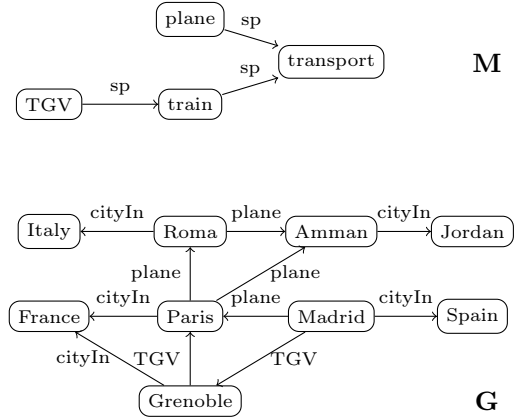
\end{example}

\subsubsection{RDF semantics}

The formal semantics of RDF expresses the conditions under which an RDF graph describes a particular world, i.e., an interpretation is a model for the graph \cite{hayes2004a}. 
The usual notions of validity, satisfiability and consequence are entirely determined by these conditions.

\begin{definition}[RDF Interpretation] Let $V \subseteq (\mathcal U \cup \mathcal L)$ be a vocabulary. An \emph{RDF interpretation} of $V$ is a tuple $I =\langle I_R, I_P, I_{EXT},\iota\rangle$ such that:
\begin{itemize*}
    \item $I_R$ is a set of \emph{resources} that contains $V \cap \mathcal L$;
    \item $I_P \subseteq I_R$ is a set of properties;
    \item $I_{EXT}: I_P \rightarrow 2^{I_R \times I_R}$ associates to each property a set of pairs of resources called  the \emph{extension} of the property;
    \item the interpretation function $\iota: V \rightarrow I_R$ associates to each name in $V$ a resource of $I_R$, such that if $v \in \mathcal L$, then $\iota(v) = v$.
\end{itemize*}
\end{definition}

\begin{definition}[RDF model] Let $V \subseteq \mathcal V$ be a vocabulary, and $G$ be an RDF graph such that $voc(G) \subseteq V$. An RDF interpretation $I = \langle I_R, I_P, I_{EXT}, \iota\rangle$ of $V$ is an \emph{RDF model} of $G$ iff there exists a mapping $\iota': \mathcal T(G) \rightarrow I_R$ that extends $\iota$, i.e., $t \in V \cap \mathcal T(G) \Rightarrow \iota'(t) = \iota(t)$, such that for each triple $\langle s, p, o\rangle\in G$, $\iota'(p) \in I_P$ and $\langle\iota'(s), \iota'(o)\rangle \in I_{EXT}(\iota'(p))$. The mapping $\iota'$ is called a \emph{proof} of $G$ in $I$.
\end{definition}

Consequence (or entailment) is defined in the standard way:

\begin{definition}[RDF entailment]\label{def:rdfentail}
A graph $G$ \emph{RDF-entails} a graph $P$ (denoted by $G \models_{RDF} P$) if and only if each RDF model of $G$ is also an RDF model of $P$.
\end{definition}

The entailment of RDF graphs (respectively, GRDF graphs) can be characterized in terms of subset for the case of graphs without variables (respectively, in terms of homomorphism when the graphs have variables).

\begin{proposition}[RDF, GRDF entailment \cite{hayes2004a}]\label{def:rdfentail}
$G \models_{RDF} P$ iff $\forall\langle s,p,o \rangle \in P$, then $\langle s,p,o \rangle \in G$. 
An RDF graph $G$ \emph{RDF-entails} a GRDF graph $P$ iff there exists a mapping $\sigma: \mathcal{T}(P) \rightarrow \mathcal{T}(G)$ such that $\forall\langle s,p,o \rangle \in P$, $\langle \sigma(s),\sigma(p),\sigma(o) \rangle \in G$.
\end{proposition}



\subsection{SPARQL}

SPARQL is the RDF query language developed by the W3C \cite{prudhommeaux2008a}.
SPARQL query answering is characterized by defining a mapping (shortened here as ``map'') from the query to the RDF graph to be queried.

We define in the following subsections the syntax and the semantics of SPARQL.
For a complete description of SPARQL, the reader is referred to the SPARQL specification \cite{prudhommeaux2008a} or to \cite{perez2009a,polleres2007a} for its formal semantics.
Unless stated otherwise, the report concentrates on SPARQL 1.0, but some features considered in the presented languages are planned to be integrated in SPARQL 1.1.

\subsubsection{SPARQL syntax}\label{sec:sparqlquerysyntax}

The basic building blocks of SPARQL queries are \emph{graph patterns} which are shared by all SPARQL query forms. 
Informally, a \emph{graph pattern} can be a triple pattern, i.e., a GRDF triple, a basic graph pattern, i.e., a set of triple patterns such as a GRDF graph, the union (UNION) of graph patterns, an optional (OPT) graph pattern, or a constraint (FILTER). 

\begin{definition}[SPARQL graph pattern] A \emph{SPARQL graph pattern} is defined inductively in the following way:
\begin{itemize*}
    \item every GRDF graph is a SPARQL graph pattern;
    \item if $P$, $P'$ are SPARQL graph patterns and $K$ is a SPARQL constraint, then ($P$ {\tt AND} $P'$), ($P$ {\tt UNION} $P'$), ($P$ {\tt OPT} $P'$), and ($P$ {\tt FILTER} $K$) are SPARQL graph patterns.
\end{itemize*}
\end{definition}

A SPARQL constraint $K$ is a boolean expression involving terms from ($\mathcal{V} \cup \mathcal{B}$), e.g., a numeric test. 
We do not specify these expressions further (see \cite{munoz2009a} for a more complete treatment).

A SPARQL SELECT query is of the form {\tt SELECT} $\vec B$ {\tt FROM} $u$ {\tt WHERE} $P$ where $u$ is the URI of an RDF graph $G$, $P$ is a SPARQL graph pattern and $\vec B$ is a tuple of variables appearing in $P$. 
Intuitively, such a query asks for the assignments of the variables in $\vec B$ such that, under these assignments, $P$ is entailed by the graph identified by $u$.

\begin{example}[Query]\label{ex:query}
The following query searches, in the regulatory network of Figure~\ref{fig:genegraph}, a gene $?x$ which inhibits a product that regulates a product that $?x$ promotes, and returns these three entities:
\begin{smalltt}
SELECT ?x, ?y, ?z
FROM ...
WHERE \{
  ?x rn:inhibits ?y
  ?x rn:promotes ?z
  ?y rn:regulates ?z
  ?x rdf:type rn:gene.
  \}
\end{smalltt}
\end{example}

\begin{example}[Query]\label{ex:query}
The following query searches in the RDF graph of Figure~\ref{fig:rdfgraph} if there exists a direct plane between a city in France and a city in Jordan:
\begin{smalltt}
SELECT ?city1 ?city2
FROM ...
WHERE
  ?city1 plane ?city2 .
  ?city1 cityIn France .
  ?city2 cityIn Jordan .
\end{smalltt}
\end{example}

SPARQL provides several other query forms that can be used for formatting query results: 
CONSTRUCT which can be used for building an RDF graph from the set of answers, 
ASK which returns {\sc true} if there is an answer to a given query and {\sc false} otherwise, 
and DESCRIBE which can be used for describing a resource RDF graph.
We concentrate on the SELECT query form and modify SPARQL basic graph patterns, leaving the remainder of the query forms unchanged.

\subsubsection{SPARQL semantics}\label{sec:sparqlquerysemantics}

In the following, we characterize query answering with SPARQL as done in \cite{perez2009a}. 
The approach relies upon the correspondence between GRDF entailment and maps from RDF graph of the query graph patterns to the RDF knowledge base.
SPARQL query constructs are defined through algebraic operations on maps: assignments from a set of variables to terms that preserve names.

\begin{definition}[Map]\label{def:map}
Let $V_1 \subseteq \mathcal{T}$, and $V_2 \subseteq \mathcal{T}$ be two sets of terms. 
A map from $V_1$ to $V_2$ is a function $\sigma: V_1 \rightarrow V_2$ such that $\forall x\in ( V_1 \cap \mathcal{V})$, $\sigma(x)=x$.
\end{definition}

\mparagraph{Operations on maps}
If $\sigma$ is a map, then the domain of $\sigma$, denoted by $dom(\sigma)$, is the subset of $\mathcal{T}$ on which $\sigma$ is defined. 
The restriction of $\sigma$ to a set of terms $X$ is defined by $\sigma |_{X}=\{\langle x, y\rangle\in\sigma |~x\in X\}$ and 
the completion of $\sigma$ to a set of terms $X$ is defined by $\sigma |^{X}=\sigma\cup\{\langle x, {\tt null}\rangle |~x\in X\text{ and }x\notin dom(\sigma)\}$\footnote{The {\tt null} symbol is used for denoting the NULL values introduced by the OPTIONAL clause.}.

If $P$ is a graph pattern, then $\mathcal{B}(P)$ is the set of variables occurring in $P$ and $\sigma(P)$ is the graph pattern obtained by the substitution of $\sigma(b)$ to each variable $b \in \mathcal{B}(P)$. 
Two maps $\sigma_1$ and $ \sigma_2$ are \emph{compatible} when $\forall x \in dom(\sigma_1) \cap dom(\sigma_2)$, $\sigma_1(x)=\sigma_2(x)$. 
Otherwise, they are said to be incompatible and this is denoted by $\sigma_1 \bot \sigma_2$. If $\sigma_1$ and $\sigma_2$ are two compatible maps, then we denote by $\sigma = \sigma_1 \oplus \sigma_2: T_1 \cup T_2 \rightarrow \mathcal T$ the map defined by $\forall x \in T_1, \sigma(x) = \sigma_1(x)$ and $\forall x \in T_2, \sigma(x) = \sigma_2(x)$. 
The \emph{join} and \emph{difference} of two sets of maps $\Omega_1$ and $ \Omega_2$ are defined as follows \cite{perez2009a}: 
\begin{itemize*}
\item \emph{(join)} $\Omega_1 \Join \Omega_2= \{\sigma_1 \oplus \sigma_2  \ | \  \sigma_1 \in \Omega_1, \sigma_2 \in\Omega_2$ are compatible$\}$;
\item \emph{(difference)} $\Omega_1 \setminus \Omega_2 = \{\sigma_1 \in \Omega_1 \ | \ \forall \sigma_2 \in \Omega_2, \sigma_1$ and $\sigma_2$ are not compatible$\}$. 
\end{itemize*}

The answers to a basic graph pattern query are those maps which warrant the entailment of the graph pattern by the queried graph. 
In the case of SPARQL, this entailment relation is RDF-entailment. 
Answers to compound graph patterns are obtained through the operations on maps.


\begin{definition}[Answers to compound graph patterns]\label{def:sparqlsynanswer}
Let $\models_{RDF}$ be the RDF entailment relation on basic graph patterns, $P$, $P'$ be SPARQL graph patterns, $K$ be a SPARQL constraint, and $G$ be an RDF graph. The set $\mathcal S(P, G)$ of \emph{answers} to $P$ in $G$ is the set of maps from $\mathcal{B}(P)$ to $\mathcal{T}(G)$ defined inductively in the following way:
\begin{align*}
\begin{split}
\mathcal{S}(P,G) &= \{\sigma|_{\mathcal{B}(P)}|~G \models_{RDF} \sigma(P)\}  \text{~~~~if $P$ is a basic graph pattern}
\end{split}\\
\mathcal S((P~\texttt{AND}~P'), G) &= \mathcal S(P, G) \Join \mathcal S(P',G)\\
\mathcal S(P~\texttt{UNION}~P', G) &= \mathcal S(P, G) \cup \mathcal S(P',G)\\
\begin{split}
\mathcal S(P~\texttt{OPT}~P', G) &= (\mathcal S(P, G) \Join \mathcal S(P',G)) \cup (\mathcal S(P, G) \setminus \mathcal{S}(P', G))
\end{split}\\
\mathcal S(P~\texttt{FILTER}~K, G) &= \{\sigma \in \mathcal S(P, G) \ | \ \sigma(K) = \top\}
\end{align*}
\end{definition}
The symbol $\models_{RDF}$ is used to denote the RDF entailment relation on basic graph patterns and we will use simply $\models$ when it is clear from the context. 
Moreover, the conditions $K$ are interpreted as boolean functions from the terms they involve. 
Hence, $\sigma(K) = \top$ means that this function is evaluated to true once the variables in $K$ are substituted by $\sigma$. 
If not all variables of $K$ are bound, then $\sigma(K)\neq\top$. 
One particular operator that can be used in SPARQL filter conditions is ``$bound(?x)$''. 
This operator returns true if the variable $?x$ is bound and in this case $\sigma(K)$ is not true whenever a variable is not bound.




As usual for this kind of query language, an answer to a query is an assignment of the distinguished variables (those variables in the \texttt{SELECT} part of the query). 
Such an assignment is a map from variables in the query to nodes of the graph. 
The defined answers may assign only one part of the variables, those sufficient to prove entailment. 
The answers are these assignments extended to all distinguished variables.


\begin{definition}[Answers to a SPARQL query]\label{def:sparqlanswer}
Let $\texttt{SELECT~} \vec{B} \texttt{~FROM~} u \texttt{~WHERE~} P$ be a SPARQL query, $G$ be the RDF graph identified by the URI $u$, and $\mathcal{S}(P,G)$ be the set of answers to $P$ in $G$, 
then the \emph{answers} $\mathcal{A}(\vec{B}, G, P)$ to the query are the restriction and completion to $\vec{B}$ of answers to $P$ in $G$, i.e.,
$$\mathcal{A}(\vec{B}, G, P)=\{ \sigma|_{\vec{B}}^{\vec{B}}|~\sigma\in \mathcal{S}(P,G)\}.$$
\end{definition}

The completion to null does not prevent that blank nodes remain in answers: null values only replace unmatched variables. 
\cite{polleres2007a} defines a different semantics for the join operation when the maps contain null values. This semantics could be adopted instead without changing the remainder of this paper.

\begin{example}[Query evaluation]\label{ex:eval}
The evaluation of the query of Example~\ref{ex:query} against the RDF graph of Example~\ref{ex:rdf} returns only one answer:
$$\langle \smtt{dm:bcd}, \smtt{dm:tll}, \smtt{dm:Kr}\rangle$$
\end{example}

In the following, we use an alternate characterization of SPARQL query answering that relies upon the correspondence between GRDF entailment and maps from the query graph patterns to the RDF graph \cite{perez2009a}. 
The previous definition can be rewritten in a more semantic style by extending entailment to compound graph patterns modulo a map $\sigma$ \cite{alkhateeb2009a}.

\begin{definition}[Compound graph pattern entailment]\label{def:sparqlsemanswer}
Let $\models$ be an entailment relation on basic graph patterns, $P$, $P'$ be SPARQL graph patterns, $K$ be a SPARQL constraint, and $G$ be an RDF graph, 
\emph{graph pattern entailment} by an RDF graph modulo a map $\sigma$ is defined inductively by:
\begin{align*}
G\models \sigma(P~\texttt{AND}~P') &\text{ iff } G\models \sigma(P)\text{ and }G\models \sigma(P')\\
G\models \sigma(P~\texttt{UNION}~P') &\text{ iff } G\models \sigma(P)\text{ or }G\models \sigma(P')\\
\begin{split}
G\models \sigma(P~\texttt{OPT}~P') &\text{ iff } G\models \sigma(P)\text{ and } [G\models \sigma(P') \text{ or } \forall \sigma'; G\models \sigma'(P'), \sigma\bot \sigma']
\end{split}\\
G\models \sigma(P~\texttt{FILTER}~K) &\text{ iff } G\models \sigma(P)\text{ and }\sigma(K)=\top
\end{align*}
\end{definition}

The following proposition is an equivalent characterization of SPARQL answers, closer to a semantic definition.

\begin{proposition}[Answer to a SPARQL query \cite{alkhateeb2009a}]\label{prop:sparql}
Let $\texttt{SELECT~}\vec{B}$ $\texttt{FROM~}u$ $\texttt{WHERE~}P$ be a SPARQL query with $P$ a SPARQL graph pattern and $G$ be the (G)RDF graph identified by the URI $u$, 
then the set of answers to this query is 
$$\mathcal{A}(\vec{B}, G, P)=\{\sigma|_{\vec{B}}^{\vec{B}}|~G\models_{(G)RDF} \sigma(P) \}.$$
\end{proposition}

The proof given in \cite{alkhateeb2009a} starts from a slightly different definition than Definition~\ref{def:sparqlsynanswer} based on RDF-entailment.
However, in the case of RDF entailment of a ground triple, these are equivalent thanks to the interpolation lemma \cite{hayes2004a}.

In order to evaluate the complexity of query answering, we use the following problem, usually named query evaluation but better named \textsc{answer checking}:

\noindent \textbf{Problem:} \textsc{$\mathcal{A}$-Answer checking}\\
\textbf{Input:} an RDF graph $G$, a SPARQL graph pattern $P$, a tuple of variables $\vec{B}$, and a map $\sigma$.\\
\textbf{Question:} Does $\sigma\in\mathcal{A}(\vec{B}, G, P)$?\\

This problem has usually the same complexity as checking if an answer exists. 
For SPARQL, the problem has been shown to be PSPACE-complete.

\begin{proposition}[Complexity of \textsc{$\mathcal{A}$-Answer checking} \cite{perez2009a}]\label{prop:SPARQLcmpl}
\textsc{$\mathcal{A}$-Answer checking} is 
PSPACE-complete.
\end{proposition} 

The complexity of checking RDF-entailment and GDRF-entailment is NP-complete \cite{gutierrez2004a}. 
This means that \textsc{$\mathcal{A}$-Answer checking} when queries are reduced to basic graph patterns is NP-complete. 
In fact, the addition of AND, FILTER, and UNION does not increase complexity which remains NP-complete.
This complexity comes from the addition of the OPT construction \cite{perez2009a}.

Hence, for every language in which entailment is NP-complete, used as a basic graph pattern language, the problem for this language will have the same complexity since Definition~\ref{def:sparqlsemanswer} shows the independence of subquery evaluation.

\section{Querying RDF modulo RDF Schema}\label{sec:modulo}

RDF Schema (or RDFS) \cite{rdfvocabulary} together with OWL \cite{owl} are formal logics recommended by the W3C for defining the vocabulary used in RDF graphs. 

RDF Schema further constrains RDF interpretations. Thus RDF graphs when interpreted under a schema may have less models, and thus more consequences. This provides more answers to a query relying on RDFS semantics. 
However, SPARQL query answers (under RDF entailment) ignore this semantics.

\begin{example}[RDF and RDFS entailment]\label{ex:rdfent}
For instance, in RDF, it is possible to deduce $\langle$\smtt{rn:inhibits} \smtt{rdf:type} \smtt{rdf:Property}$\rangle$ from $\langle$\smtt{dm:hb} \smtt{rn:inhibits} \smtt{dm:kni}$\rangle$; in RDFS, one can deduce $\langle$\smtt{dm:hb} \smtt{rn:regulates} \smtt{dm:kni}$\rangle$ from \{$\langle$\smtt{dm:hb} \smtt{rn:inhibits} \smtt{dm:kni}$\rangle$, $\langle$\smtt{rn:inhibits} \smtt{rdfs:subPropertyOf} \smtt{rn:regulates}$\rangle$\}.
\end{example}

\subsection{RDF Schema}\label{sec:rdfs}

This section focusses on RDF and RDFS as extensions of the Simple RDF language presented in Section~\ref{sec:grdf}. 
Both extensions are defined in the same way:
\begin{itemize*}
	\item They consider a particular set of urirefs of the vocabulary prefixed by \smtt{rdf:} and \smtt{rdfs:}, respectively.
	\item They add additional constraints on the resources associated to these terms in interpretations.
\end{itemize*}
We present them together below.

\subsubsection{RDF Schema vocabulary}

There exists a set of reserved words, the RDF Schema vocabulary \cite{rdfvocabulary}, designed to describe relationships between resources like classes, e.g., \smtt{dm:gap rdfs:subClassOf rn:gene}, and relationships between properties, e.g., \smtt{rn:inhibits rdfs:subPropertyOf rn:regulates}. The RDF Schema vocabulary is given in Table~\ref{tab:rdfsvocabulary} as it appears in \cite{hayes2004a}. The shortcuts that we will use for each of the terms are given in brackets.
We use $rdfsV$ to denote the RDF Schema vocabulary.

\begin{table}
\begin{center}
\begin{tabular}{lll}
rdfs:domain[$\smtt{dom}$]   & rdfs:Container[$\smtt{cont}$]  & rdf:Bag[$\smtt{bag}$]  \\
rdfs:range[$\smtt{range}$]  & rdfs:isDefinedBy[$\smtt{isDefined}$] & rdf:Seq[$\smtt{seq}$] \\
rdfs:Class[$\smtt{class}$]  & rdfs:Literal[$\smtt{literal}$]  & rdf:List[$\smtt{list}$] \\
rdf:value[$\smtt{value}$]   & rdfs:subClassOf[$\smtt{sc}$]  & rdf:Alt[$\smtt{alt}$] \\
rdfs:label[$\smtt{label}$]  & rdfs:subPropertyOf[$\smtt{sp}$]  & rdf: 1[$\smtt{1}$]\\
rdf:nil[$\smtt{nil}$]       & rdfs:comment[$\smtt{comment}$]  & \ldots\\
rdf:type[$\smtt{type}$]     &  rdf:predicate[$\smtt{pred}$] & rdf: i[$\smtt{2}$]\\
rdf:object[$\smtt{obj}$]    & rdf:Statement[$\smtt{stat}$] & rdf:first[$\smtt{first}$]\\
rdf:subject[$\smtt{subj}$]      & rdfs:seeAlso[$\smtt{seeAlso}$]  &  rdf:rest[$\smtt{rest}$]\\
rdf:Property[$\smtt{prop}$]       & rdfs:Datatype[$\smtt{datatype}$] & rdfs:member[$\smtt{member}$]\\
rdfs:Resource[$\smtt{res}$]     & rdf:XMLLiteral[$\smtt{xmlLit}$] \\
\multicolumn{3}{c}{rdfs:ContainerMembershipProperty[$\smtt{contMP}$]} \\
\end{tabular}		
\end{center}
\caption{The RDF Schema vocabulary.}	\label{tab:rdfsvocabulary}
\end{table}

We will consider here a core subset of RDFS, $\rho$df \cite{munoz2009a}, also called the description logic fragment of RDFS \cite{cali2009a}. It contains the following vocabulary:
\begin{center}
$\rho$df = $\{\smtt{sc}, \smtt{sp}, \smtt{type}, \smtt{dom}, \smtt{range}\}$
\end{center}

\begin{example}[RDFS]\label{ex:rdfs}
The RDF graph of Example~\ref{ex:rdf} can be expressed in the context of an RDF Schema which provides more information about the vocabulary that it  uses. It specifies subtypes of genes and subtypes of regulation relations.\\

\noindent\begin{tabular}{ll}
\begin{minipage}{.65\textwidth}
\begin{smalltt}
dm:maternal rdfs:subClassOf rn:gene.
dm:gap rdfs:subClassOf rn:gene.
rn:regulates rdfs:domain rn:gene.
rn:regulates rdfs:range rn:gene.
rn:inhibits rdfs:subPropertyOf rn:regulates.
rn:promotes rdfs:subPropertyOf rn:regulates.
rn:inhibits\_translation rdfs:subPropertyOf rn:inhibits.
rn:inhibits\_transcription rdfs:subPropertyOf rn:inhibits.
\end{smalltt}
\end{minipage}
&
\begin{minipage}{.3\textwidth}
\begin{smalltt}
dm:kni rdf:type dm:gap.
dm:hb rdf:type dm:gap.
dm:Kr rdf:type dm:gap.
dm:tll rdf:type dm:gap.
dm:bcd rdf:type dm:maternal.
dm:cad rdf:type dm:maternal.
\end{smalltt}
\end{minipage}
\end{tabular}
\end{example}

\subsubsection{RDF Schema semantics}

In addition to the usual interpretation mapping, a special mapping is used in RDFS interpretations for interpreting the set of classes which is a subset of $I_R$.

\begin{definition}[RDFS interpretation]
 An \emph{RDFS interpretation} of a vocabulary $V$ is a tuple $\langle I_R,$ $I_P,$ $Class,$ $I_{EXT},$ $I_{CEXT},$ $Lit,$ $\iota\rangle$ such that: 
\begin{itemize*}
\item $\langle I_R, I_P, I_{EXT}, \iota\rangle$ is an RDF interpretation;
\item $Class \subseteq I_R$ is a distinguished subset of $I_R$ identifying if a resource denotes a class of resources; 
\item $I_{CEXT}: Class \rightarrow 2^{I_R}$ is a mapping that assigns a set of resources to every resource denoting a class; 
\item $Lit \subseteq I_R$ is the set of literal values, $Lit$ contains all plain literals in $\mathcal {L} \cap V$.	
\end{itemize*}
\end{definition}

Specific conditions are added to the resources associated to terms of RDFS vocabularies in an RDFS interpretation to be an RDFS model of an RDFS graph. These conditions include the satisfaction of the RDF Schema axiomatic triples as appearing in the normative semantics of RDF \cite{hayes2004a}.

\begin{definition}[RDFS Model]\label{def:rdfsmodel}
Let $G$ be an RDFS graph, and $I=\langle I_R,$ $I_P,$ $Class,$ $I_{EXT},$ $I_{CEXT},$ $Lit,$ $\iota \rangle$ be an RDFS interpretation of a vocabulary $V \subseteq rdfsV \cup \mathcal{V}$ such that $\mathcal{V}(G) \subseteq V$.    Then $I$ is an \emph{RDFS model} of $G$ if and only if $I$ satisfies the following conditions:
\begin{enumerate*}
	\item Simple semantics:
				\begin{enumerate}[a)]
					\item there exists an extension $\iota'$ of $\iota$ to $\mathcal B(G)$ such that for each triple $\langle s, p, o\rangle$ of $G$, $\iota'(p) \in I_P$ and $\langle \iota'(s), \iota'(o)\rangle \in I_{EXT}(\iota'(p))$.
				\end{enumerate}

   \item RDF semantics:
				\begin{enumerate}[a)]
					\item $x \in I_P \Leftrightarrow \langle x, \iota'(\smtt{prop})\rangle \in  I_{EXT}(\iota'(\smtt{type}))$.
          \item If $\ell \in term(G)$ is a typed XML literal with lexical form $w$, then $\iota'(\ell)$ is the XML literal value of $w$, $\iota'(\ell) \in Lit$, and $\langle \iota'(\ell), \iota'(\smtt{xmlLit})\rangle \in I_{EXT}(\iota'(\smtt{type}))$.
				\end{enumerate}
   \item RDFS Classes:
   			\begin{enumerate}[a)]
					\item $x \in I_R$, $x \in I_{CEXT}(\iota'(\smtt{res}))$.
					\item $x \in Class$, $x \in I_{CEXT}(\iota'(\smtt{class}))$.
					\item $x \in Lit$, $x \in I_{CEXT}(\iota'(\smtt{literal}))$.
				\end{enumerate}

   \item RDFS Subproperty:
        \begin{enumerate}[a)]
					\item $I_{EXT}(\iota'(\smtt{sp}))$ is transitive and reflexive over $I_P$.
					\item if $\langle x, y\rangle \in I_{EXT}(\iota'(\smtt{sp}))$ then $x, y \in I_P$ and $I_{EXT}(x) \subseteq I_{EXT}(y)$.
				\end{enumerate}
   \item  RDFS Subclass:
        \begin{enumerate}[a)]
					\item $I_{EXT}(\iota'(\smtt{sc}))$ is transitive and reflexive over $Class$.
					\item $\langle x, y\rangle \in I_{EXT}(\iota'(\smtt{sc}))$, then $x, y \in Class$ and $I_{CEXT}(x)$ $\subseteq I_{CEXT}(y)$.
				\end{enumerate}
   \item RDFS Typing:
        \begin{enumerate}[a)]
					\item $x \in I_{CEXT}(y)$, $(x,y) \in I_{EXT}(\iota'(\smtt{type}))$.
					\item if $\langle x, y\rangle  \in I_{EXT}(\iota'(\smtt{dom}))$ and $\langle u, v\rangle \in I_{EXT}(x)$ then $u \in I_{CEXT}(y)$.
					\item if $\langle x, y\rangle  \in I_{EXT}(\iota'(\smtt{range}))$ and $\langle u, v\rangle \in I_{EXT}(x)$ then $v \in I_{CEXT}(y)$.
				\end{enumerate}
   \item RDFS Additionals:
        \begin{enumerate}[a)]
					\item if $x \in Class$ then $\langle x, \iota'(\smtt{res})\rangle \in I_{EXT}(\iota'(\smtt{sc}))$.
          \item if $x \in I_{CEXT}(\iota'(\smtt{datatype}))$ then $\langle x, \iota'(\smtt{literal})\rangle \in I_{EXT}(\iota'(\smtt{sc}))$.
          \item if $x \in I_{CEXT}(\iota'(\smtt{contMP}))$ then $\langle x, \iota'(\smtt{member})\rangle \in I_{EXT}(\iota'(\smtt{sp}))$.
				\end{enumerate}
\end{enumerate*}
\end{definition}

\begin{definition}[RDFS entailment]\label{def:rdfsentai}
Let $G$ and $P$ be two RDFS graphs, then $G$ \emph{RDFS-entails} $P$ (denoted by $G \models_{\text{RDFS}} P$) if and only if every RDFS model of $G$ is also an RDFS model of $P$.
\end{definition}

\begin{table*}
\centering
\begin{tabular}{ccc}
$\dfrac{\langle p,\smtt{sp},q\rangle ~~\langle q,\smtt{sp},r\rangle}{\langle p,\smtt{sp},r\rangle}$ & 
$\dfrac{\langle A,\smtt{sc},B\rangle~~\langle B,\smtt{sc},C\rangle}{\langle A,\smtt{sc},C\rangle}$ & 
$\dfrac{\langle p,\smtt{dom},A\rangle~~\langle x,p,y\rangle} {\langle x,\smtt{type},A\rangle}$ \vspace{2mm}\\ 
$\dfrac{\langle p,\smtt{sp},q\rangle~~\langle x,p,y\rangle} {\langle x,q,y\rangle}$  &
$\dfrac{\langle A,\smtt{sc},B\rangle~~\langle x,\smtt{type},A\rangle} {\langle x,\smtt{type},B\rangle}$ & 
$\dfrac{\langle p,\smtt{range},A\rangle~~\langle x,p,y\rangle} {\langle y,\smtt{type},A\rangle}$ 
\end{tabular}
\caption{RDFS inference rules (from \protect\cite{munoz2009a}).}\label{tab:rdfs-rules}
\end{table*}

\cite{munoz2009a} has introduced the reflexive relaxed semantics for RDFS in which \smtt{rdfs:subClassOf} and \smtt{rdfs:subPropertyOf} do not have to be reflexive. The entailment relation with this semantics is noted $\models_{\text{RDFS}}^{\text{nrx}}$.

The reflexive relaxed semantics does not change much RDFS. 
Indeed, from the standard (reflexive) semantics, we can deduce that any class (respectively, property) is a subclass (respectively, subproperty) of itself. 
For instance, $\langle\smtt{dm:hb}~\smtt{rn:inhibits}~\smtt{dm:kni}\rangle$ only entails $\langle\smtt{rn:inhibits}~\smtt{sp}~\smtt{rn:inhibits}\rangle$ and variations of this triple in which occurrences of \smtt{rn:inhibits} are replaced by variables.
The reflexivity requirement only entails reflectivity assertions which do not interact with other triples unless constraints are added to the \smtt{rdfs:subPropertyOf} or \smtt{rdfs:subClassOf} predicates. 
Therefore, it is assumed that elements of the RDFS vocabulary appear only in the predicate position.
We will call \emph{genuine}, RDFS graphs which do not constrain the elements of the $\rho$df vocabulary (and thus these two predicates), and restrict us to querying genuine RDFS graphs.

However, when issuing queries involving these relations, e.g., with a graph pattern like $\langle$\smtt{?x}~\smtt{sp}~\smtt{?y}$\rangle$, all properties in the graph will be answers. 
Since this would clutter results, we assume, as done in \cite{munoz2009a}, that queries use the reflexive relaxed semantics. 
It is easy to recover the standard semantics by providing the additional triples when \smtt{sp} or \smtt{sc} are queried. 

In the following, we use the closure graph of an RDF graph $G$, denoted by $closure(G)$, which is defined by the graph obtained by saturating $G$ with all triples that can be deduced using rules of Table~\ref{tab:rdfs-rules}.

The SPARQL specification \cite{prudhommeaux2008a} introduces the notion of entailment regimes.
These regimes contain several components (query language, graph language, inconsistency handling) \cite{glimm2010a}. 
We concentrate here on the definition of answers which, in particular, replace simple RDF entailment for answering queries. 
It is possible to define answers to SPARQL queries modulo RDF Schema, by using RDFS entailment as the entailment regime.

\begin{definition}[Answers to a SPARQL query modulo RDF Schema]\label{def:answersparqlmodulo}
Let $\texttt{SELECT~}\vec{B}$ $\texttt{FROM~}u$ $\texttt{WHERE~}P$ be a SPARQL query with $P$ a GRDF graph and $G$ be the RDFS graph identified by the URI $u$, 
then the set of \emph{answers} to this query \emph{modulo RDF Schema} is:
$$\mathcal{A}^{\#}(\vec{B}, G, P)=\{\sigma|_{\vec{B}}^{\vec{B}}| G\models_{RDFS}^{\text{nrx}} \sigma(P) \}$$
\end{definition}

This definition is justified by the analogy between RDF entailment and RDFS entailment in the definition of answers to queries (see \cite{alkhateeb2009a}). It does not fully correspond to the RDFS entailment regime defined in \cite{glimm2010a} since we do not restrict the answer set to be finite. This is not strictly necessary, however, the same restrictions as in \cite{glimm2010a} can be applied.

The problem is to specify a query engine that can find such answers.

\subsection{Querying against ter Horst closure}\label{sec:rdfscloure}

One possible approach for querying an RDFS graph $G$ in a sound and complete way is by computing the closure graph of $G$, i.e., the graph obtained by saturating $G$ with all information that can be deduced using a set of predefined rules called RDFS rules, before evaluating the query over the closure graph. 

\begin{definition}[RDFS closure]\label{def:rdfsclosure}
Let $G$ be an RDFS graph on an RDFS vocabulary $V$. 
The \emph{RDFS closure} of $G$, denoted $\hat{G}$, is the smallest set of triple containing $G$ and satisfying the following constraints:

\begin{tabular}{rp{12cm}}
[RDF1] & all RDF axiomatic triples \cite{hayes2004a} are in $\hat{G}$;\\

[RDF2] & if $\langle s,p,o \rangle\in\hat{G}$, 
then $\langle p,\smtt{type},\smtt{prop}\rangle\in\hat{G}$;\\

[RDF3] & if $\langle s,p,\ell \rangle\in\hat{G}$, where $\ell$ is an $\smtt{xmlLit}$ typed literal and the lexical representation $s$ is a well-formed  XML literal,\\
& then $\langle s,p,xml(s) \rangle\in\hat{G}$ and $\langle xml(s), \smtt{type},$ $\smtt{xmlLit} \rangle\in\hat{G}$;\\

[RDFS 1] & all RDFS axiomatic triples \cite{hayes2004a} are in $\hat{G}$;\\

[RDFS 6] & if $\langle  a, \smtt{dom}, x  \rangle\in\hat{G}$  and $\langle u, a ,y  \rangle\in\hat{G}$, 
then  $\langle u,$ $\smtt{type}, x  \rangle\in\hat{G}$;\\	

[RDFS 7] & if $\langle  a ,\smtt{range}, x  \rangle\in\hat{G}$ and $\langle u, a ,v  \rangle\in\hat{G}$, 
then  $\langle v,$ $ \smtt{type},x  \rangle\in\hat{G}$;\\

[RDFS 8a] & if $\langle x, \smtt{type}, \smtt{prop}  \rangle\in\hat{G}$, 
then $\langle x ,\smtt{sp}, x  \rangle\in\hat{G}$;\\

[RDFS 8b] & if $\langle  x, \smtt{sp}, y  \rangle\in\hat{G}$ and $\langle y ,\smtt{sp}, z  \rangle\in\hat{G}$, 
then $\langle x,$ $ \smtt{sp},z\rangle\in\hat{G}$;\\

[RDFS 9] & if $\langle  a, \smtt{sp}, b  \rangle\in\hat{G}$ and $\langle x, a, y  \rangle\in\hat{G}$, 
then $\langle x,$ $b, y  \rangle\in\hat{G}$;\\

[RDFS 10] & if $\langle x, \smtt{type} ,\smtt{class} \rangle\in\hat{G}$, 
then $\langle x, \smtt{sc}, \smtt{res} \rangle\in\hat{G}$;\\

[RDFS 11] & if $\langle u, \smtt{sc}, x  \rangle\in\hat{G}$ and $ \langle y, \smtt{type}, u \rangle\in\hat{G}$, 
then $\langle y, $ $\smtt{type},x  \rangle\in\hat{G}$;\\


[RDFS 12a] & if $\langle  x, \smtt{type}, \smtt{class} \rangle\in\hat{G}$, then $\langle x, \smtt{sc}, x  \rangle\in\hat{G}$;\\

[RDFS 12b] & if $\langle  x, \smtt{sc}, y  \rangle\in\hat{G}$ and $ \langle y ,\smtt{sc}, z  \rangle\in\hat{G}$, 
then $\langle x , $ $\smtt{sc},z \rangle\in\hat{G}$;\\

[RDFS 13]  & if $\langle x, \smtt{type}, \smtt{contMP} \rangle\in\hat{G}$,
 then $\langle x, \smtt{prop},$ $\smtt{member}\rangle\in\hat{G}$;\\

[RDFS 14] & if $\langle x ,\smtt{type}, \smtt{datatype} \rangle\in\hat{G}$,
then $\langle x, \smtt{sc},$ $ \smtt{literal} \rangle\in\hat{G}$.\\
\end{tabular}
\end{definition}

It is easy to show that this closure always exists and can be obtained by turning the constraints into rules, thus defining a closure operation.

\begin{example}[RDFS Closure]\label{ex:closure}
The RDFS closure of the RDF graph of Example~\ref{ex:rdf} augmented by the RDFS triples of Example~\ref{ex:rdfs} contains, in particular, the following assertions:
\begin{smalltt}
dm:bcd rn:inhibits dm:cad. // [RDFS 9]
dm:hb rn:regulates dm:kni. // [RDFS 9]
dm:hb type rn:gene.        // [RDFS 6]
\end{smalltt}
\end{example}

Because of axiomatic triples, this closure may be infinite, but
a finite and polynomial closure, called \emph{partial closure}, has been proposed independently in \cite{baget2003a} and \cite{terhorst2005a}. 

\begin{definition}[Partial RDFS closure]\label{def:partialrdfsclosure}
Let $G$ and $H$ be two RDFS graphs on an RDFS vocabulary $V$,
the \emph{partial RDFS closure} of $G$ given $H$, denoted $\hat{G}\backslash H$, is obtained in the following way:
\begin{enumerate*}
\item  let $k$ be the maximum of $i$'s such that \smtt{rdf\!:\!\_i} is a term of $G$ or of $H$;
\item replace the rule [RDF 1] by the rule \\
\begin{tabular}{rp{12cm}}
[RDF 1P] & add all RDF axiomatic triples \cite{hayes2004a} except those that use \smtt{rdf\!:\!\_i} with $i > k$;\\
\end{tabular}\\
\noindent In the same way, replace the rule [RDFS 1] by the rule \\
\begin{tabular}{rp{12cm}}
[RDFS 1P] & add all RDFS axiomatic triples except those that use \smtt{rdf\!:\!\_i} with $i > k$;\\
\end{tabular}
\item apply the modified rules.
\end{enumerate*}
\end{definition}

Applying the partial closure to an RDFS graph permits to reduce RDFS entailment to simple RDF entailment. 

\begin{proposition}[Completeness of partial RDFS closure \cite{hayes2004a}]\label{thm:rdfsentailment} 
Let $G$ be a satisfiable RDFS graph and $H$ an RDFS graph, then $G \models_{RDFS} H$ if and only if $(\hat{G}\backslash H) \models_{RDF} H$. 
\end{proposition}

The completeness does not hold if $G$ is not satisfiable because in such a case, any graph $H$ is a consequence of $G$ and $\models_{RDF}$ does not reflect this (no RDF graph can be inconsistent).
In case $G$ is unsatisfiable, the RDFS entailment regime allows for raising an error \cite{glimm2010a}.
An RDFS graph can be unsatisfiable only if it contains datatype conflicts \cite{terhorst2005a} which can be found in polynomial time.

Since queries must adopt the reflexive relaxed semantics, we have to further restrict this closure. It can be obtained by suppressing constraints RDFS8a and RDFS12a from the closure operation. We denote the partial non reflexive closure  $\hat{G}\backslash\backslash H$.

\begin{proposition}[Completeness of partial non reflexive RDFS closure]\label{thm:rdfnrxsentailment} 
Let $G$ be a satisfiable genuine RDFS graph and $H$ an RDFS graph, then $G \models_{RDFS}^{\text{nrx}}  H$ if and only if $(\hat{G}\backslash\backslash H) \models_{RDF} H$. 
\end{proposition}

This has the following corollary: 

\begin{corollary}\label{thm:answersparqlmodulo}
$$\mathcal{A}^{\#}(\vec{B}, G, P)=\mathcal{A}(\vec{B}, \hat{G}\backslash\backslash P, P)$$
\end{corollary}

\begin{example}[SPARQL evaluation modulo RDFS]\label{ex:rdfseval}
If the query of Example~\ref{ex:query} is evaluated against the RDFS closure of Example~\ref{ex:closure},
it will return the three expected answers:
\begin{align*}
\{ & \langle \smtt{dm:hb}, \smtt{dm:kni}, \smtt{dm:Kr}\rangle\\
& \langle \smtt{dm:bcd}, \smtt{dm:tll}, \smtt{dm:Kr}\rangle\\
& \langle \smtt{dm:bcd}, \smtt{dm:cad}, \smtt{dm:kni}\rangle \}
\end{align*}
This can be obtained easily from the simple graph of Example~\ref{ex:rdf} augmented by the triples of Example~\ref{ex:closure}.
\end{example}

This shows the correctness and completeness of the closure approach.
This approach has several drawbacks which limit its use: 
It still tends to generate a very large graph which makes it not very convenient, especially if the transformation has to be made on the fly, i.e., when the query is evaluated; 
It takes time proportional to $|H| \times |G|^2$ in the worst case \cite{munoz2009a}; 
Moreover, it is not applicable if one has no access to the graph to be queried but only to a SPARQL endpoint. 
In this case, it is not possible to compute the closure graph.

Since the complexity of the partial closure has been shown to be polynomial \cite{terhorst2005a}, \textsc{$\mathcal{A}^{\#}$-Answer checking} remains PSPACE-complete.
 
\begin{proposition}[Complexity of \textsc{$\mathcal{A}^{\#}$-Answer checking}]\label{prop:terHorstcmpl}
\textsc{$\mathcal{A}^{\#}$-Answer checking} is PSPACE-complete.
\end{proposition} 

\section{The PSPARQL query language}\label{sec:psparql}

In order to address the problems raised by querying RDF graphs modulo RDF Schemas, we first present the PSPARQL query language for RDF which we introduced in \cite{alkhateeb2009a}. 
Unlike SPARQL, PSPARQL can express queries with regular path expressions instead of relations and variables on the edges. 
For instance, it allows for finding all maternal genes regulated by the \smtt{dm:bcd} gene through an arbitrary sequence of inhibitions.

We will, in the next section, show how the additional expressive power provided by PSPARQL can be used for answering queries modulo RDF Schema.

The added expressiveness of PSPARQL is achieved by extending SPARQL graph patterns, hence any SPARQL query is a PSPARQL query. SPARQL graph patterns, based on GRDF graphs, are replaced by PRDF graphs that are introduced below through their syntax (\S\ref{sec:crdf-syntax}) and semantics (\S\ref{sec:crdf-semantics}).
They are then naturally replaced within the PSPARQL context (\S\ref{sec:psparqlc}).

\subsection{PRDF syntax} \label{sec:crdf-syntax}

Let $\Sigma$ be an alphabet. A \emph{language} over $\Sigma$ is a subset of $\Sigma^*$: its elements are sequences of elements of $\Sigma$ called \emph{words}. A (non empty) word $\langle a_1, \ldots,$ $a_k\rangle$ is denoted $a_1 \cdot \ldots \cdot a_k$. If $A = a_1 \cdot \ldots \cdot a_k$ and $B = b_1 \cdot \ldots \cdot b_q$ are two words over $\Sigma$, then $A \cdot B$ is the word over $\Sigma$ defined by $A \cdot B = a_1 \cdot \ldots \cdot a_k \cdot b_1 \cdot \ldots \cdot b_q$.

\begin{definition}[Regular expression pattern] \label{def:rep}
   Let $\Sigma$ be an alphabet, $X$ be a set of variables, the set $\mathcal{R}(\Sigma, X)$ of \emph{regular expression patterns} is inductively defined by: 
\begin{itemize*} 
  \item $\forall a \in \Sigma$, $a\in\mathcal{R}(\Sigma, X)$ and $!a\in\mathcal{R}(\Sigma, X)$; 
  \item $\forall x \in X$, $x\in\mathcal{R}(\Sigma, X)$; 
  \item $\epsilon \in \mathcal{R}(\Sigma, X)$;
  \item If  $A \in \mathcal{R}(\Sigma, X)$ and $B \in \mathcal{R}(\Sigma, X)$ then $A|B$, $A \cdot B$, $A^*$, $A^+ \in \mathcal{R}(\Sigma, X)$.
\end{itemize*} 
\end{definition} 

A regular expression over ($\mathcal{U}, \mathcal{B}$) can be used to define a language over the alphabet made of $\mathcal{U} \cup \mathcal{B}$.
PRDF graphs are GRDF graphs where predicates in the triples are regular expression patterns constructed over the set of URI references and the set of variables. 

\begin{definition}[PRDF graph] A \emph{PRDF triple} is an element of $\mathcal{U}\cup\mathcal{B}$ $\times$ $\mathcal{R}(\mathcal{U}, \mathcal{B})$ $\times$ $\mathcal{T}$. A PRDF graph is a set of PRDF triples.
\end{definition}

Hence all GRDF graphs are PRDF graphs.

\begin{example}[PRDF Graph]\label{ex:prdf}
PRDF graphs can express interesting features of regulatory networks. For instance, one may observe that \smtt{dm:bcd} promotes \smtt{dm:Kr} without knowing if this action is direct or indirect. Hence, this can be expressed by \smtt{dm:bcd rn:promotes+ dm:Kr}.

A generalized version of the graph pattern of the query of Example~\ref{ex:query} can be expressed by:
\begin{smalltt}
  ?x rn:inhibits.rn:regulates* ?z.
  ?x rn:promotes+ ?z.
  ?x rdf:type rn:gene.
\end{smalltt}
\end{example}

\subsection{PRDF semantics}\label{sec:crdf-semantics}

To be able to define models of PRDF graphs, we have first to express path semantics within RDF semantics to support regular expressions. 

\begin{definition}[Support of a regular expression]\label{def:support}
 Let   $I=\langle I_R, I_P, I_{EXT}, \iota\rangle$ be an  interpretation of a vocabulary $V=\mathcal{U}\cup \mathcal{L}$,  $\iota'$ be an extension of $\iota$ to $B\subseteq \mathcal{B}$, and $R\in \mathcal{R}(\mathcal{U}, B)$, a pair $\langle x,y \rangle$ of ($I_R \times I_R$) \emph{supports}  $R$ in $\iota'$ if and only if one of the two following conditions are satisfied: 
\begin{enumerate} [(i)]
        \item  the empty word $\epsilon \in L^*(R)$ and $x=y$; 
        \item   there exists a word of length $n\geq 1$ $w=w_1 \cdot \ldots \cdot w_n$ where $w  \in L^*(R)$ and 
a sequence of resources of $I_R$ $x=r_0,\ldots ,r_n=y$  such that $\langle r_{i-1},r_{i}\rangle \in I_{EXT}(\iota'(w_i))$, $1 \leq i \leq n$. 
\end{enumerate} 
\end{definition} 

Instead of considering paths in RDF graphs, this definition considers paths in the interpretations of PRDF graphs, i.e., paths are now relating resources. This is used in the following definition of PRDF models in which it replaces the direct correspondences that exist in RDF between a relation and its interpretation, by a correspondence between a regular expression and a sequence of relation interpretations. This allows for matching regular expressions, e.g., $r^{+}$, with variable length paths.

\begin{definition}[Model of a PRDF graph]\label{def:modelprdf} 
 Let $G$ be a PRDF graph, and $I=\langle I_R, I_P, I_{EXT}, \iota\rangle$ be an interpretation of a vocabulary $V \supseteq\mathcal{V}(G)$, $I$ is a \emph{PRDF model} of $G$ if and only if there exists an  extension $\iota'$ of $\iota$ to  $\mathcal{B}(G)$  such that  for every triple $\langle s,R,o \rangle \in G$, $\langle \iota'(s), \iota'(o)\rangle$ supports $R$ in $\iota'$.           
\end{definition} 

This definition extends the definition of RDF models, and they are equivalent when all regular expression patterns $R$ are reduced to atomic terms, i.e., urirefs or variables. PRDF entailment is defined as usual:

\begin{definition}[PRDF entailment]\label{def:prdfentail} 
Let $P$ and $G$ be two PRDF graphs, $G$ \emph{PRDF-entails} $P$ (noted $G\models_{PRDF} P$) if and only if all models of $G$ are models of $P$.           
\end{definition} 

It is possible to define the interpretation of a regular expression evaluation as those pairs of resources which support the expression in all models.

\begin{definition}[Regular expression interpretation]\label{def:prdfentail}
The interpretation $\[[R\]]_{G}$ of a regular expression $R$ in a PRDF graph $G$ is the set of nodes which satisfy the regular expression in all models of the graph:
$$\[[R\]]_{G}=\{\langle x, y\rangle |~ \forall I \text{ PRDF model of }G, \langle x, y\rangle\text{ supports }R\text{ in }I\}$$
\end{definition} 

It is thus possible to formulate the regular expression evaluation problem:

\noindent \textbf{Problem:} Regular expression evaluation\\
\textbf{Input:} An RDF graph $G$, a regular expression $R$, and a pair $\langle a, b \rangle$\\
\textbf{Question:} Does $\langle a, b \rangle \in \[[R\]]_{G}$?\\

We will use this same problem with different type of regular expressions.


\subsection{PSPARQL}\label{sec:psparqlc}

PSPARQL is an extension of SPARQL introducing the use of paths in SPARQL graph patterns. 
PSPARQL graph patterns are built on top of PRDF in the same way as SPARQL is built on top of GRDF. 

\begin{definition}[PSPARQL graph pattern] A \emph{PSPARQL graph pattern} is defined inductively by:
\begin{itemize*}
    \item every PRDF graph is a PSPARQL graph pattern;
    \item if $P_1$ and $P_2$ are two PSPARQL graph patterns and $K$ is a SPARQL constraint, then ($P_1$ {\tt AND} $P_2$), ($P_1$ {\tt UNION} $P_2$), ($P_1$ {\tt OPT} $P_2$), and ($P_1$ {\tt FILTER} $K$) are PSPARQL graph patterns.
\end{itemize*}
\end{definition}

A PSPARQL query for the select form is {\tt SELECT} $\vec{B}$ {\tt FROM} $u$ {\tt WHERE} $P$ such that $P$ is a PSPARQL graph pattern.


Analogously to SPARQL, the set of answers to a PSPARQL query is defined inductively from the set of maps of the PRDF graphs of the query into the RDF knowledge base. 
The definition of an answer to a PSPARQL query is thus identical to Definition~\ref{def:sparqlsemanswer}, but it uses PRDF entailment.

\begin{definition}[Answers to a PSPARQL query]\label{def:answerpsparql}
Let $\texttt{SELECT~} \vec{B}$ $\texttt{~FROM~} u$ $\texttt{~WHERE~} P$ be a PSPARQL query with $P$ a PRDF graph pattern, and $G$ be the RDF graph identified by the URI $u$, 
then the set of \emph{answers} to this query is:
$$\mathcal{A}^{\star}(\vec{B}, G, P)=\{\sigma|_{\vec{B}}^{\vec{B}}| G\models_{PRDF} \sigma(P) \}$$
\end{definition}

\subsection{SPARQL queries modulo RDFS with PSPARQL}\label{sec:cprdfs}

To overcome the limitations of previous approaches when querying RDF graphs modulo an RDF Schema, we provide a new approach which rewrites a SPARQL query into a PSPARQL 
query using a set of rules, and then evaluates the transformed query over the graph to be queried. 
In particular, we show that every SPARQL query $Q$ to evaluate over an RDFS graph $G$ can be transformed into a PSPARQL query $\tau(Q)$ such that evaluating $Q$ over $\hat{G}$, the closure graph of $G$, is equivalent to evaluating $\tau(Q)$ over $G$. 

The query rewriting approach is similar in spirit to the query rewriting methods using a set of views \cite{papakonstantinou1999,calvanese2000a,grahne2003}.
In contrast to these methods, our approach uses the data contained in the graph, i.e., the rules are inferred from RDFS entailment rules.
We define a rewriting function $\tau$  from RDF graph patterns to PRDF graph patterns through a set of rewriting rules over triples (which naturally extends to basic graph patterns and queries). $\tau(Q)$ is obtained from $Q$ by applying the possible rule(s) to each triple in $Q$.

\begin{definition}[Basic RDFS graph pattern expansion]
Given an RDF triple $t$, the \emph{RDFS expansion} of $t$ is a finite PSPARQL graph pattern $\tau(t)$ defined as:
\begin{align*}
\tau(\langle s,\smtt{sc},o\rangle) =& \{\langle s,\smtt{sc}^+,o\rangle\}\\
\tau(\langle s,\smtt{sp},o\rangle) =& \{\langle s,\smtt{sp}^+,o\rangle\}\\
\tau(\langle s,p,o\rangle) =& \{\langle s,?x,o\rangle, \langle ?x,\smtt{sp}^*,p\rangle\} (p\not\in\{\smtt{sc},\smtt{sp},\smtt{type}\})\\
\tau(\langle s,\smtt{type},o\rangle) =& \{\langle s,\smtt{type} \cdot \smtt{sc}^*,o\rangle\} \\
                        \texttt{UNION}~& \{\langle s,?p,?y\rangle, \langle ?p,\smtt{sp}^*\cdot\smtt{dom}\cdot \smtt{sc}^*,o\rangle\}\\
                        \texttt{UNION}~& \{\langle ?y,?p,s\rangle , \langle ?p,\smtt{sp}^*\cdot\smtt{range}\cdot \smtt{sc}^*,o\rangle\}       
\end{align*}
\end{definition}

The first rule handles the transitive semantics of the subclass relation. Finding the subclasses of a given class can be achieved by navigating all its direct subclasses.
The second rule handles similarly the transitive semantics of the subproperty relation. 
The third rule tells that the subject-object pairs occurring in the subproperties of a given property are inherited to it. 
Finally, the fourth rule expresses that the instance mapped to $s$ has for type the class mapped to $o$ (we use the word ``mapped'' since $s$ and/or $o$ can be variables) if one of the following conditions holds:
\begin{enumerate*}
\item the instance mapped to $s$ has for type one of the subclasses of the class mapped to $o$ by following the subclass relationship zero or several times. The zero times is used since $s$ can be directly of type $o$;
\item there exists a property of which $s$ is subject and such that the instances appearing as a subject must have for type one of the subclasses of the class mapped to $o$;
\item there exists a property of which $s$ is object and such that the instances appearing as an object must have for type one of the subclasses of the class mapped to $o$.
\end{enumerate*}

The latter rule takes advantage of a feature of PSPARQL:
the ability to have variables in predicates.

\begin{example}[PSPARQL Query]\label{ex:psparql}
The result of transforming the query of Example~\ref{ex:query} with $\tau$ is:
\begin{smalltt}
SELECT ?x, ?y, ?z
FROM ...
WHERE \{
  ?x ?r ?y. ?r sp* rn:inhibits.
  ?y ?t ?z.?t sp* rn:regulates.
  ?x ?s ?z. ?s sp* rn:promotes.
  (	?x rdf:type.sc* rn:gene. 
  UNION 
  	\{ ?x ?u ?v. ?u sp*.dom.sc* rn:gene.\}
  UNION
  	\{ ?v ?u ?x. ?u sp*.range.sc* rn:gene.\}
  )
  \}
\end{smalltt}
This query provides the correct set of answers for the RDF graph of Example~\ref{ex:rdf} modulo the RDF Schema of Example~\ref{ex:rdfs} (given in Example~\ref{ex:rdfseval}).
\end{example}

Any SPARQL query can be answered modulo an RDF Schema by rewriting the resulting query in PSPARQL and evaluating the PSPARQL query. 

\begin{proposition}[Answers to a SPARQL query modulo RDF Schema by PSPARQL]\label{thm:cpsparqlrdfsans}
$$\mathcal{A}^{\#}(\vec{B}, G, P)=\mathcal{A}^{\star}(\vec{B}, G, \tau(P))$$
\end{proposition}


This transformation does not increase the complexity of PSPARQL which is the same as the one of SPARQL:

\begin{proposition}[Complexity of \textsc{$\mathcal{A}^{\star}$-Answer checking}]\label{prop:PSPARQLcmpl}
\textsc{$\mathcal{A}^{\star}$-Answer checking} is PSPACE-complete.
\end{proposition} 

\section{nSPARQL and NSPARQL}\label{sec:nsparql}

An alternative to PSPARQL was proposed with the nSPARQL language, a simple query language based on nested regular expressions for navigating RDF graphs \cite{perez2010a}.
We present it as well as NSPARQL, an extension more comparable to PSPARQL.

\subsection{nSPARQL syntax}

\begin{definition}[Regular expression]\label{def:re}
A \emph{regular expression} is an expression built from the following grammar:
\begin{center}
$re$ ::= axis $|$ axis\emph{::}a $|$ re $|$ re\emph{/}re $|$ re\emph{|}re $|$ re$^*$
\end{center}
\noindent with $a\in \mathcal{U}$ and $axis\in$\{\smtt{self}, \smtt{next}, \smtt{next$^{-1}$}, \smtt{edge}, \smtt{edge$^{-1}$}, \smtt{node}, \smtt{node$^{-1}$} $\}$.
\end{definition}
In the following, we use the positive closure of a path expression $R$ denoted by $R^+$ and defined as $R^+=R/R^*$. 

Regarding the precedence among the regular expression operators, it is as follows: *, /, then |. Parentheses may be used for breaking precedence rules.

The model underlying nSPARQL is that of XPath which navigates within XML structures. 
Hence, the axis denotes the type of node object which is selected at each step, respectively, the current node (\smtt{self} or \smtt{self}$^{-1}$), the nodes reachable through an outbound triple (\smtt{next}), the nodes that can reach the current node through an incident triple (\smtt{next}$^{-1}$), the properties of outbound triples (\smtt{edge}), the properties of incident triples (\smtt{edge}$^{-1}$), the object of a predicate (\smtt{node}) and the predicate of an object (\smtt{node}$^{-1}$). 
This is illustrated by Figure~\ref{fig:axis}.

\begin{figure}[htp]
 \begin{center}
\begin{tikzpicture}
\tikzstyle{every node}=[rounded corners,x=40pt,y=45pt]

\draw (1,0) node[draw] (t1) {\scriptsize{subject}};
\draw (5,0) node[draw] (t3) {\scriptsize{objet}};
\draw[->] (t1) -- node[above,draw] (t2) {\scriptsize{predicate}} (t3);

\draw[dotted,->] (t1.west) .. controls +(-1,1)  and + (-1,-1) .. node[below,near end] {\smtt{self}} (t1.west);
\draw[dotted,->] (t1) .. controls +(1,2)  and + (-1,2) .. node[below] {\smtt{next}} (t3);
\draw[dotted,->] (t3) .. controls +(-1,-2)  and + (1,-2) .. node[above] {\smtt{next}$^{-1}$} (t1);
\draw[dotted,->] (t1) .. controls +(1,.5)  and + (-1,.5) .. node[above] {\smtt{edge}} (t2);
\draw[dotted,->] (t2) .. controls +(-1,-.5)  and + (1,-.5) .. node[below] {\smtt{edge}$^{-1}$} (t1);
\draw[dotted,->] (t2) .. controls +(1,.5)  and + (-1,.5) .. node[above] {\smtt{node}} (t3);
\draw[dotted,->] (t3) .. controls +(-1,-.5)  and + (1,-.5) .. node[below] {\smtt{node}$^{-1}$} (t2);
\draw[dotted,->] (t3.east) .. controls +(1,1)  and + (1,-1) .. node[below,near end] {\smtt{self}$^{-1}$} (t3.east);

\end{tikzpicture}						
\end{center}
\caption{nSPARQL axes.}\label{fig:axis}
\end{figure}
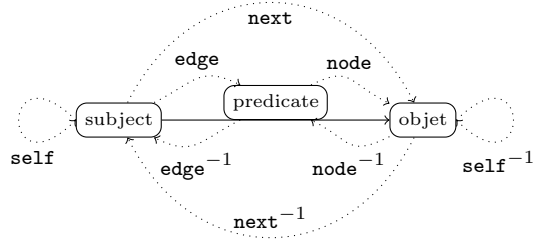

\begin{definition}[Nested regular expression]\label{def:nre}
A \emph{nested regular expression} is an expression built from the following grammar (with $a\in \mathcal{U}$):
\begin{center}
$nre$ ::= axis $|$ axis\emph{::}a $|$ axis\emph{::}$[$nre$]$ $|$ nre $|$ nre\emph{/}nre $|$ nre\emph{|}nre $|$ nre$^*$
\end{center}
\end{definition}

Contrary to simple regular expressions, nested regular expressions may constrain nodes to satisfy additional secondary paths.

Nested regular expressions are used in triple patterns in predicate position, to define nSPARQL triple patterns.   

\begin{definition}[nSPARQL triple pattern]\label{def:npathtriple}
An \emph{nSPARQL triple pattern} is a triple $\langle s, p, o\rangle$ such that $s\in\mathcal{T}$, $o\in\mathcal{T}$ and $p$ is a nested regular expression.
\end{definition} 

\begin{example}[nSPARQL triple pattern]\label{ex:nsparqlpattern}
Assume that one wants to retrieve the pairs of cities such that there is a way of traveling by any transportation mean. 
The following nSPARQL pattern expresses this query:
$$P= \langle ?city_1, (next::[(next::sp)^*/self::transport])^+, ?city_2 \rangle$$
This pattern expresses a sequence of properties such that each property (predicate) is a sub-property of the property "transport". 
\end{example}

\begin{example}[nSPARQL triple pattern]\label{ex:nsparql}
In the context of molecular biology, nSPARQL expressions may be very useful. For instance, 
part of the graph patterns used for the query of Example~\ref{ex:query} can be expressed by:
\begin{smalltt}
?x next::rn:inhibits / next::rn:regulates ?z
\end{smalltt}
which finds all pairs of nodes such that the first one inhibits a regulator of the second one.
It can be further enhanced by using transitive closure:
\begin{smalltt}
?x next::rn:inhibits / next::rn:regulates+ ?z
\end{smalltt}
expressing that we want a path between two nodes going through a first inhibition and then an arbitrary non null number of regulatory steps (\smtt{+} is the usual notation such that \smtt{a+} corresponds to \smtt{a/a*}).
Nested expressions allow for going further by constraining any step in the path. So,
\begin{smalltt}
?x next::rn:inhibits[ next::rdf:type/self::dm:gap ] / next::rn:regulates+ ?z
\end{smalltt}
requires, in addition, the second node in the path to be a gap gene.
\end{example}

From nSPARQL triple patterns, it is also possible to create a query language from nSPARQL triple patterns.
As in SPARQL, a set of nSPARQL triple patterns is called an nSPARQL basic graph pattern. 
nSPARQL graph patterns may be defined in the usual way, i.e., by replacing triple patterns by nSPARQL triple patterns.

\begin{definition}[nSPARQL graph pattern]\label{def:nsparqlgp} An \emph{nSPARQL graph pattern} is defined inductively by:
\begin{itemize*}
    \item every nSPARQL triple pattern is an nSPARQL graph pattern;
    \item if $P_1$ and $P_2$ are two nSPARQL graph patterns and $K$ is a SPARQL constraint, then ($P_1$ {\tt AND} $P_2$), ($P_1$ {\tt UNION} $P_2$), ($P_1$ {\tt OPT} $P_2$), and ($P_1$ {\tt FILTER} $K$) are nSPARQL graph patterns.
\end{itemize*}
\end{definition}

For time complexity reasons the designers of the nSPARQL language choose to define a more restricted language than SPARQL \cite{perez2010a}.
Contrary to SPARQL queries, nSPARQL queries are reduced to nSPARQL graph patterns, constructed from nSPARQL triple patterns, plus SPARQL operators AND, UNION, FILTER, and OPT.
They do not allow for the projection operator (SELECT). 
This prevents, when checking answers, that uncontrolled variables have to be evaluated.

\subsection{nSPARQL semantics}

In order to define the semantics of nSPARQL, we need to know the semantics of nested regular expressions \cite{perez2010a}.
Here we depart from the semantics given in \cite{perez2010a} by adding a third variable in the interpretation whose sole purpose is to compact the set of rules. 
Both definitions are equivalent.
\begin{definition}[Nested path interpretation]\label{def:nresemantics}
Given a nested path $p$ and an RDF graph $G$, the interpretation of $p$ in $G$ (denoted $\[[p\]]_G$) is defined by:
\begin{align*}
\[[\smtt{self}\]]_{G} &=~\{\langle x, x, x\rangle; x\in \mathcal{T}\}\\ 
\[[\smtt{next}\]]_{G} &=~\{\langle x, y, z\rangle; \exists z; \langle x, z, y\rangle\in G\}\\
\[[\smtt{edge}\]]_{G} &=~\{\langle x, z, z\rangle; \exists z; \langle x, z, y\rangle\in G\}\\ 
\[[\smtt{node}\]]_{G} &=~\{\langle z, y, z\rangle; \exists z; \langle x, z, y\rangle\in G\}\\
\[[nre\textup{\smtt{::}}\smtt{a}\]]_{G} &=~\{\langle x, y, a\rangle\in \[[nre\]]_{G}\}\\
\[[nre^{\textup{\smtt{-1}}}\]]_{G} &=~\{\langle y, x, z\rangle; \langle x, y, z\rangle\in \[[nre\]]_{G}\}\\
\[[nre_1\textup{\smtt{[}}nre_2\textup{\smtt{]}}\]]_{G} &=~\{\langle x, y, z\rangle\in \[[nre_1\]]_{G};  \exists w, k; \langle z, w, k\rangle\in \[[nre_2\]]_{G}\}\\
\[[nre_1\textup{\smtt{/}}nre_2\]]_{G} &=~\{\langle x, y, z\rangle; \langle x, w, k\rangle\in \[[nre_1\]]_{G}  \&~\langle w, y, z\rangle\in \[[nre_2\]]_{G}\}\\
\[[nre_1\textup{\smtt{|}}nre_2\]]_{G} &=~\[[nre_1\]]_{G}\cup\[[nre_2\]]_{G}\\
\[[nre\texttt{*}\]]_{G} &=~\[[\smtt{self}\]]_{G}\cup \[[nre\]]_{G}\cup \[[nre\textup{\smtt{/}}nre\]]_{G}  \cup\[[nre\textup{\smtt{/}}nre\textup{\smtt{/}}nre\]]_{G}\cup\dots
\end{align*}
\end{definition}

The evaluation of a nested regular expression $R$ over an RDF graph $G$ is defined as the sets of pairs $\langle a, b\rangle$ of nodes in $G$, such that $b$ is reachable from $a$ in $G$ by following a path that conforms to $R$.
We will write $\langle x, y\rangle\in\[[R\]]_G$ as a shortcut for $\exists z$ such that $\langle x, y, z\rangle\in\[[R\]]_G$.

\begin{definition}[Satisfaction of a nSPARQL triple pattern]\label{def:npatans}
Given a basic nested path graph pattern $\langle s, p, o\rangle$ and an RDF graph $G$, $\langle s, o\rangle$ satisfies $p$ in $G$ (denoted $G\models_{nSPARQL} \langle s, p, o\rangle$) if and only if $\exists \sigma;$ $\langle \sigma(s), \sigma(o)\rangle\in\[[p\]]_G$
\end{definition}

This nested regular expression evaluation problem is solved efficiently through an effective procedure provided in \cite{perez2010a}.

\begin{theorem}[Complexity of nested regular expression evaluation \cite{perez2010a}]
The evaluation problem for a nested regular expression $R$ over an RDF graph $G$ can be solved in time $O(|G|\times|R|)$.
\end{theorem}

Answers to nSPARQL queries follow the same definition as for SPARQL (Definition~\ref{def:sparqlsemanswer}) but with maps satisfying nSPARQL triple patterns.

\begin{definition}[Evaluation of a nSPARQL triple pattern]\label{def:nsparql-triple-evaluation}
The evaluation of a nSPARQL triple pattern $\langle x,R,y\rangle$ over an RDF graph $G$ is: 
$$
\[[\langle x,R,y\rangle\]]_{G} = \{\sigma |  dom(\sigma) = \{x, y\} \cap {\mathcal B}
\text{ and } \langle \sigma(x), \sigma(y)\rangle \in \[[R\]]_{G}\}$$
\end{definition}

\cite{perez2010a} shows that avoiding the projection operator (SELECT), keeps the complexity of nSPARQL \emph{basic}, i.e., conjunctive, graph pattern evaluation to polynomial and mentions that adding projection would make it NP-complete.

Clearly, nSPARQL is a good navigational language, but there still are useful queries that could not be expressed. 
For example, it cannot be used to find nodes connected with transportation mean that is not a bus or transportation means belonging to Air France, i.e., containing the URI of the company. 

\subsection{NSPARQL}\label{sec:NSPARQL}

It is also possible to create a query language from nSPARQL triple patterns by simply replacing SPARQL triple patterns by nSPARQL triple patterns. 
We call such a language NSPARQL for differentiating it from the original nSPARQL \cite{perez2010a}. 
However, the merits of the approach are directly inherited from the original nSPARQL. 

A NSPARQL query for the select form is {\tt SELECT} $\vec{B}$ {\tt FROM} $u$ {\tt WHERE} $P$ such that $P$ is a nSPARQL graph pattern (see Definition~\ref{def:nsparqlgp}).
Hence, NSPARQL graph patterns are built on top of nSPARQL in the same way as SPARQL is built on top of GRDF and PSPARQL is built on top of PRDF. 


Answers to NSPARQL queries are based on the extension of $\models_{nSPARQL}$ nSPARQL graph patterns following Definition~\ref{def:sparqlsemanswer} using $\models_{nSPARQL}$ as the entailment relation.

\begin{definition}[Answers to an NSPARQL query]\label{def:nsparqlanswer}
Let $\texttt{SELECT~}\vec{B}$ $\texttt{FROM~}u$ $\texttt{WHERE~}P$ be a NSPARQL query with $P$ a nSPARQL graph pattern and $G$ be the (G)RDF graph identified by the URI $u$, 
then the set of \emph{answers} to this query is:
$$\mathcal{A}^{o}(\vec{B}, G, P)=\{\sigma|_{\vec{B}}^{\vec{B}}|~G\models_{nSPARQL} \sigma(P) \}.$$
\end{definition}

The complexity of NSPARQL query evaluation is likely to be PSPACE-complete as SPARQL (see \S\ref{sec:sparqlquerysemantics}).

\subsection{SPARQL queries modulo RDFS with nSPARQL}\label{sec:rdfs-nsparql}

\begin{definition}\label{def:rdfs-nsparql-triple-evaluation}
The evaluation of an nSPARQL triple pattern $\langle x,R,y \rangle$ over an RDF graph $G$ modulo RDFS is defined as
$$\[[\langle x,R,y \rangle\]]_{G}^{rdfs} = \[[\langle x,R,y \rangle\]]_{closure(G)}$$
\end{definition}

\begin{definition}[Answers to an nSPARQL basic graph pattern modulo RDFS]\label{def:rdfs-nsparqlanswer}
Let $P$ be a basic nSPARQL graph pattern and $G$ be  an RDF graph, then the set of \emph{answers} to $P$ over $G$ modulo RDFS is:
$$\mathcal{S}^{o}(P,G)= \bigcap_{t\in P} \[[t\]]^{rdfs}_{G}$$
\end{definition}

As presented in \cite{perez2010a}, nSPARQL can evaluate queries with respect to RDFS by transforming the queries with rules  \cite{munoz2009a}:
\begin{align*}
\phi(sc) & = (\smtt{next::sc})\smtt{+}\\ 
\phi(sp) & = (\smtt{next::sp})\smtt{+}\\ 
\phi(dom) & = \smtt{next::dom}\\ 
\phi(range) & = \smtt{next::range}\\ 
\phi(type) & = \smtt{next::type}/\smtt{next::sc}\smtt{*}\\
& ~~~ |\smtt{edge}/\smtt{next::sp}\smtt{*}/\smtt{next::dom}/\smtt{next::sc}\smtt{*}\\
& ~~~ |\smtt{node}^{-1}/\smtt{next::sp}\smtt{*}/\smtt{next::range}\\
& ~~~  /\smtt{next::sc}\smtt{*}\\
\phi(p) & = \smtt{next[(next::sp)*/self::p]} ~~(p\not\in\{\smtt{sp}, \smtt{sc}, \smtt{type}, \smtt{dom}, \smtt{range}\})
\end{align*}

\begin{example}[nSPARQL evaluation modulo RDFS]\label{ex:nsparqlquery}
The following nSPARQL graph pattern could be used as a query to retrieve the set of pairs of cities connected by a sequence of transportation means such that one city is from France and the other city is from Jordan: 
\begin{align*}
\{\langle ?city_1, (next::&transport)^+, ?city_2 \rangle, \\
\langle ?city_1, next::&cityIn, France \rangle,  \\
\langle ?city_2, next::&cityIn, Jordan \rangle\}
\end{align*}
When evaluating this graph pattern against the RDF graph of Figure~\ref{fig:rdfgraph}, it returns the empty set since there is no explicit ``transport'' property between the two queried cities. 
However, considering the RDFS semantics, it should return the following set of pairs:
$$\{\langle ?city_1 \leftarrow Paris, ?city_2 \leftarrow Amman\rangle, \langle ?city_1 \leftarrow Grenoble, ?city_2 \leftarrow Amman\rangle \}$$

To answer the above graph pattern considering RDFS semantics, it could be transformed to the following nSPARQL graph pattern:
\begin{align*}
\{\langle ?city_1, (next::&[(next::sp)^*/self::transport])^+, ?city_2 \rangle, \\
\langle ?city_1, next::&cityIn, France \rangle,  \\
\langle ?city_2, next::&cityIn, Jordan \rangle\}
\end{align*}
\end{example}

This encoding is correct and complete with respect to RDFS entailment.

\begin{theorem}[Completeness of $\phi$ on triples \cite{perez2010a} (Theorem~3)]\label{prop:nsparqlcompl}
Let $\langle x,p,y \rangle$ be a SPARQL triple pattern with $x, y\in ({\mathcal U} \cup {\mathcal B})$ and $p \in {\mathcal U}$, then for any RDF graph $G$:
 $$\[[\langle x,p,y \rangle\]]_{G}^{rdfs} = \[[\langle x,\phi(p),y \rangle\]]_{G}$$
\end{theorem}

This results uses the natural extension of $\phi$ to nSPARQL graph patterns to answer SPARQL queries modulo RDFS.

\begin{proposition}[Completeness of $\phi$ on graph patterns \cite{perez2010a}]\label{prop:nsparqlquerycompl}
Given a basic graph pattern $P$ and an RDFS graph $G$,
$$G\models_{RDFS}^{\text{nrx}} P \text{ iff } G\models_{nSPARQL} \phi(P)$$
\end{proposition}

\subsection{SPARQL queries modulo RDFS with NSPARQL}\label{sec:nsparqlrdfs}

These results about nSPARQL can be transferred to NSPARQL query answering:

\begin{proposition}[Answers to a SPARQL query modulo RDFS by NSPARQL]\label{prop:nsparqlanswer}
Let $\texttt{SELECT~}\vec{B}$ $\texttt{FROM~}u$ $\texttt{WHERE~}P$ be a SPARQL query with $P$ a SPARQL graph pattern and $G$ the RDFS graph identified by the URI $u$, 
then the set of answers to this query is: 
$$\mathcal{A}^{\#}(\vec{B}, G, P)=\{\sigma|_{\vec{B}}^{\vec{B}}~|~\sigma\in \mathcal{S}^{o}(\phi(P),G) \}$$
\end{proposition}

\begin{example}[NSPARQL Query]\label{ex:nsparql}
The result of transforming the query of Example~\ref{ex:query} with $\phi$ is:
\begin{smalltt}
SELECT ?x, ?y, ?z
FROM ...
WHERE \{
  ?x next[(next::sp)*/self::rn:inhibits] ?y.
  ?y next[(next::sp)*/self::rn:regulates] ?z.
  ?x next[(next::sp)*/self::rn:promotes] ?z.
  ?x next::rdf:type/next::sc*
    |edge/next::sp*/next::dom/next::sc*
    |node-1/next::sp*/next::range/next::sc* rn:gene. 
    \}
\end{smalltt}
This query provides the correct set of answers for the RDF graph of Example~\ref{ex:rdf} modulo the RDF Schema of Example~\ref{ex:rdfs} (given in Example~\ref{ex:rdfseval}).
\end{example}

The main problem with NSPARQL is that, because nested regular expressions do not contain variables, it does not preserve SPARQL queries. 
Indeed, it is impossible to express a query containing the simple triple $\langle ?x~?y~?z\rangle$. This may seem like a minor problem, but in fact NSPARQL graph patterns prohibits dependencies between two variables as freely as it is possible to express in SPARQL.
For instance, querying a gene regulation network for self-regulation cycles (such that a product which inhibits another product that indirectly activates itself or which activates a product which indirectly inhibit itself), can be achieved with the following SPARQL query:
\begin{smalltt}
SELECT ?a
WHERE \{
    ?a ?p ?b.
    ?b rn:promotes ?c.
    ?c ?q ?a.
    ?a rdf:type rn:gene.
    ?b rdf:type rn:gene.
    ?c rdf:type rn:gene.
    ?p sp rn:regulates.
    ?q sp rn:regulates.
    ?p owl:inverseOf ?q.
    \}
\end{smalltt}
Such queries are representative of queries in which two different predicates (\smtt{?p} and \smtt{?q}) are dependent of each others\footnote{Here through \smtt{owl:inverseOf}, but the use of OWL vocabulary is not at stake here, it could have been \smtt{rdfs:subPropertyOf} or any other property.}. They are not expressible by a language like NSPARQL.
An nSPARQL expression can express (abstracting for the type constraints):
\begin{smalltt}
?a next[(next::sp)*/self::rn:regulates] / next[(next::sp)*/self::rn:regulates] 
        / next[(next::sp)*/self::rn:regulates] ?a
\end{smalltt}
\noindent but it cannot express the \smtt{owl:inverseof} constraints because it defines a dependency between two places in the path (this is not anymore a regular path).

The same type of query can be used for querying a banking system for detecting money laundering in which a particular amount is moving and coming back through an account through buy/sell or debit/credit operations with an intermediary transfer.

\section{cpSPARQL and CPSPARQL} \label{sec:cpsparql}

CPSPARQL has been defined for addressing two main issues. 
The first one comes from the need to extend PSPARQL and thus to allow for expressing constraints on nodes of traversed paths;
while the second one comes from the need to answer PSPARQL queries modulo RDFS so that the transformation rules could be applied to PSPARQL queries \cite{alkhateeb2008b}.

In addition to CPSPARQL, we present cpSPARQL \cite{alkhateeb2014a}, a language using CPSPARQL graph patterns in the same way as nSPARQL does.

\subsection{CPSPARQL syntax}

The notation that we use for the syntax of CPSPARQL is slightly different from the one defined in the original proposal \cite{alkhateeb2008b}. 
The original one uses \texttt{edge} and \texttt{node} constraints to express constraints on predicates (or edges) and nodes of RDF graphs, respectively. 
In this paper, we adopt the \texttt{axes} borrowed from XPath, with which the reader may be more familiar, as done for nSPARQL. 
This also will allow us to better compare cpSPARQL and nSPARQL. 
Additionally, in the original proposal, \texttt{ALL} and \texttt{EXISTS} keywords are used to express constraints on all traversed nodes or to check the existence of a node in the traversed path that satisfies the given constraint. 
We do not use these keywords in the fragment presented below since they do not add expressiveness with respect to RDFS semantics, i.e., the fragment still captures RDFS semantics.


Constraints act as filters for paths that must be traversed by constrained regular expressions and select those whose nodes satisfy encountered constraint. 

\begin{definition}[Constrained regular expression]\label{def:cre}
A \emph{constrained regular expression} is an expression built from the following grammar:
\begin{center}
$cre$ ::= axis $|$ axis\emph{::}a $|$ axis\emph{::}$[?x:\psi]$ $|$ axis\emph{::}$]?x:\psi[$
 $|$ cre $|$ cre\emph{/}cre $|$ cre\emph{|}cre $|$ cre$^*$ 
\end{center}
with $\psi$ a set of triples belonging to $\mathcal{U}\cup \mathcal{B}\cup \{?x\} \times cre \times \mathcal{T}\cup\{?x\}$ and FILTER-expressions over $\mathcal{B}\cup\{?x\}$. $\psi$ is called a CPRDF-constraint and $?x$ its head variable.
\end{definition}
Constrained regular expressions allow for constraining the item in one axis to satisfy a particular constraint, i.e., to satisfy a particular graph pattern (here an RDF graph) or filter.
We introduce the closed square brackets and open square brackets notation for distinguishing between constraints which export their variable (it may be assigned by the map) and constraints which do not export it (the variable is only notational). 
This is equivalent to the initial CPSPARQL formulation, in which the variable was always exported, since CPSPARQL can ignore such variables through projection.

We use ${\mathcal B}(R)$ for the set of variables occurring as the head variable of an open bracket constraint in $R$.

Constraint nesting is allowed because constrained regular expressions may be used in the graph pattern of another constrained regular expression as in Example~\ref{ex:constre}.

\begin{example}[Constrained regular expression]\label{ex:constre}
The following constrained regular expression could be used to find nodes connected by transportation means that are not buses:
\begin{align*}
(next::[?p: \{ &\langle ?p, (next::sp)^*, transport \rangle FILTER (?p !=bus)\}])^+
\end{align*}
\end{example}

In contrast to nested regular expressions, constrained regular expressions can apply constrains (such as SPARQL constraints) in addition to simple nested path constraints.

Constrained regular expressions are used in triple patterns, precisely in predicate position, to define CPSPARQL.   

\begin{definition}[CPSPARQL triple pattern]\label{def:cpsparqlpathtriple}
A \emph{CPSPARQL triple pattern} is a triple $\langle s, p, o\rangle$ such that $s\in\mathcal{T}$, $o\in\mathcal{T}$ and $p$ is a constrained regular expression.
\end{definition}

\begin{definition}[CPSPARQL graph pattern]
A \emph{CPSPARQL graph pattern} is defined inductively by:
\begin{itemize*}
    \item every CPSPARQL triple pattern is a CPSPARQL graph pattern;
    \item if $P_1$ and $P_2$ are two CPSPARQL graph patterns and $K$ is a SPARQL constraint, then ($P_1$ {\tt AND} $P_2$), ($P_1$ {\tt UNION} $P_2$), ($P_1$ {\tt OPT} $P_2$), and ($P_1$ {\tt FILTER} $K$) are CPSPARQL graph patterns.
\end{itemize*}
\end{definition}

\begin{example}[CPSPARQL graph pattern]\label{ex:cpsparqlgp}
The following CPSPARQL graph pattern could be used to retrieve the set of pairs of cities connected by a sequence of transportation means (which are not buses) such that one city in France and the other one in Jordan:
\begin{align*}
\{\langle ?city_1, (next&::[?p: \{ \langle ?p, (next::sp)^*, transport \rangle FILTER (?p !=bus)\}])^+, ?city_2 \rangle \\
\langle ?city_1, next&::cityIn, France \rangle  \\
\langle ?city_2, next&::cityIn, Jordan \rangle\}
\end{align*}
If open square brackets were used, this graph pattern would, in addition, bind the $?p$ variable to a matching value, i.e., the transportation means used.
\end{example}

By restricting CPRDF constraints, it is possible to define a far less expressive language. 
cpSPARQL is such a language.

\begin{definition}[cpSPARQL regular expression  \cite{alkhateeb2014a}]\label{def:cpcre}
A \emph{cpSPARQL regular expression} is an expression built from the following grammar:
\begin{align*}
cpre ::= ~& axis ~|~ axis\emph{::}a ~|~ axis\emph{::}]?x: TRUE[\\
~|~ & axis\emph{::}[?x: \{ \langle ?x, cpre, v \rangle \}\{ FILTER(?x) \} ]\\
~|~ & cpre ~|~ cpre\emph{/}cpre ~|~ cpre\emph{|}cpre ~|~ cpre^* 
\end{align*}
such that $v$ is either a distinct variable $?y$ or a constant (an element of $U\cup L$) and $FILTER(?x)$ is the usual SPARQL filter condition containing at most the variable $?x$ and $v$ if $v$ is a variable.
\end{definition}

The first specific form, with open square brackets, has been preserved so that cpSPARQL triples cover SPARQL basic graph patterns, i.e., allow for variables in predicate position.
In the other specific forms, a cpSPARQL constraint is either a cpSPARQL regular expression containing $?x$ as the only variable and/or a SPARQL FILTER constraint.
Hence, such a regular expression may have several constraints, but each constraint can only expose one variable and it cannot refer to variables defined elsewhere.
It is clear that any cpSPARQL regular expression is a constrained regular expression.

Deciding if a CPSPARQL triple is a cpSPARQL triple can be decided in linear time in the size of the regular expression used.


\begin{example}[cpSPARQL triple patterns]\label{ex:cpsparqlpattern}
The query of Example~\ref{ex:nsparqlpattern} could be expressed by the following cpSPARQL pattern:
\begin{align*}
\langle &?city_1, (next::[?p: \{ \langle ?p, (next::sp)^*, transport \rangle\}])^+, ?city_2 \rangle
\end{align*}
The constraint $\psi \!= \!?p\!:\! \{ \langle ?p, (next\!::\!sp)^*, transport \rangle\}$ is used to restrict the properties (in this pattern the constraint is applied to properties since the axis \smtt{next} is used) to be only a transportation mean.

Example~\ref{ex:constre} provides another cpSPARQL regular expression. By contrast, CPSPARQL graph patterns allow for queries like:
\begin{align*}
next::[&?p; \{\langle?p,(next::sp)*,?z\rangle, \\
& \qquad \langle?q, (next::sp)*, ?z\rangle, \\
& \qquad \langle?p, owl:inverseOf, ?q\rangle,\\
& \qquad FILTER(regex(?z, iata.orgÃ¹)\}]
\end{align*}
 which is not a cpSPARQL regular expression since it uses more than two variables.
\end{example}

It is possible to develop languages based on cpSPARQL regular expressions following what is done with constrained regular expressions.

\subsection{CPSPARQL semantics}

Intuitively, a constrained regular expression $next$::$[\psi]$ (where $\psi =?p: \{\langle ?p,$ $sp^*,$ $transport \rangle\}$) is equivalent to $next$::$p$ if $p$ satisfies the constraint $\psi$, i.e., $p$ should be a sub-property of $transport$ (when $p$ is substituted to the variable $?p$).

\begin{definition}[Satisfied constraint in an RDF graph]\label{def:satconstarint}
Let $G$ be an RDF graph, $s$ and $o$ be two nodes of $G$ and $\psi=x\!:\!C$ be a constraint, then $s$ and $o$ satisfies $\psi$ in $G$ (denoted $\langle s, o \rangle \in \[[\psi\]]_{G}$) if and only if one of the following conditions is satisfied: 
\begin{enumerate*}
\item $C$ is a triple pattern $C=\langle x,R,y \rangle$, and  $\langle x_s^x, y_o^y \rangle \in \[[R_s^x\]]_{G}$, where $K_r^z$ means that $r$ is substituted to the variable $z$ if $K=z$ or $K$ contains the variable $z$. If $z$ is a constant then $z=r$.   
	
\item $C$ is a SPARQL filter constraint and $C_{s,o}^{x,y}=\top$, where $C_{s,o}^{x,y}=\top$ means that the constraint obtained by the substitution of $s$ to each occurrence of the variable $x$ and $o$ to each occurrence of the variable $y$ in $C$ is evaluated to true\footnote{Except for the case of bound (see Definition~\ref{def:sparqlsynanswer} and the discussion after it).}.  
	
\item $C = P~FILTER~K$, then 1 and 2 should be satisfied 
\end{enumerate*}
\end{definition}

As for nested regular expressions, the evaluation of a constrained regular expression $R$ over an RDF graph $G$ is defined as a binary relation $\[[R\]]_{G}$, by a pair of nodes $\langle a, b\rangle$ such that $a$ is reachable from $b$ in $G$ by following a path that conforms to $R$. 
The following definition extends Definition~\ref{def:nresemantics} to take into account the semantics of terms with constraints.

\begin{definition}[Constrained path interpretation]\label{def:cresemantics}
Given a constrained regular expression $P$ and an RDF graph $G$, if $P$ is unconstrained then the interpretation of $P$ in $G$ (denoted $\[[P\]]_G$) is as in Definition~\ref{def:nresemantics}, otherwise the interpretation of $P$ in $G$  is defined as:
\begin{align*}
\[[self::[\psi]\]]_{G} =& ~\{\langle x, x\rangle|~ \exists z; x \in voc(G) \wedge \langle x, z \rangle \in \[[\psi\]]_{G}\}\\
\[[next::[\psi]\]]_{G} =& ~\{\langle x, y\rangle|~ \exists z,w; \langle x, z, y\rangle\in G \wedge\langle z, w \rangle \in \[[\psi\]]_{G}\}\\
\[[edge::[\psi]\]]_{G} =& ~\{\langle x, y\rangle|~ \exists z, w; \langle x,  y, z\rangle\in G \wedge \langle z, w \rangle \in \[[\psi\]]_{G}\}\\
\[[node::[\psi]\]]_{G} =& ~\{\langle x, y\rangle|~ \exists z,w; \langle z, x, y\rangle\in G \wedge \langle z, w \rangle \in \[[\psi\]]_{G}\}\\
\[[axis^{-1}::[\psi]\]]_{G} =& ~\{\langle x, y\rangle|~ \langle y, x \rangle\in\[[axis::[\psi]\]]_{G}\}
\end{align*}
\end{definition}

\begin{definition}[Answer to a CPSPARQL triple pattern]\label{def:triple-evaluation}
The evaluation of a CPSPARQL triple pattern $\langle x,R,y \rangle$ over an RDF graph $G$ is defined as the following set of maps:
$$\[[\langle x,R,y \rangle\]]_{G} = \{\sigma \ | \ dom(\sigma) = \{x, y\} \cap {\mathcal B}\cup {\mathcal B}(R) \text{ and }\langle \sigma(x), \sigma(y) \rangle \in \[[\sigma(R)\]]_{G}\}$$
such that $\sigma(R)$ is the constrained regular expression obtained by substituting the variable $?x$ appearing in a constraint with open brackets in $R$ by $\sigma(?x)$. 
\end{definition}

This semantics also applies to cpSPARQL graph patterns. 


\subsection{Evaluating cpSPARQL regular expressions}\label{sec:complexityresult}

In order to establish the complexity of cpSPARQL we follow \cite{perez2010a} to store an RDF graph as an adjacency list: every $u \in voc(G)$ is associated with a list of pairs $\alpha(u)$. For instance, if $\langle s,p,o\rangle \in G$, then $\langle$\smtt{next}::$p,o \rangle \in \alpha(s)$ and $\langle$\smtt{edge}$^{-1}$::$o,s \rangle \in \alpha(p)$. Also, $\langle$\smtt{self}::$u,u\rangle \in \alpha(u)$, for $u \in voc(G)$.
The set of terms of a constrained regular expression $R$, denoted by ${\mathcal T}(R)$, is constructed as follows:
\vspace{-.25cm}
\begin{align*}
\mathcal{T}(R) =& \{R\} \text{if $R$ is either \smtt{axis}, \smtt{axis}::a, or \smtt{axis}::$\psi$}\\
\mathcal{T}(R_1 /R_2) =&  {\mathcal T}(R_1 | R_2) ={\mathcal T}(R_1) \cup {\mathcal T}(R_2)\\
\mathcal{T}(R_1^*) =& {\mathcal T}(R_1)
\end{align*}

Let $\mathcal{A}_R=(Q,{\mathcal T}(R),s_0,F,\delta)$ be the $\epsilon-NFA$ of $R$ constructed in the usual way using the terms ${\mathcal T}(R)$, where $\delta: Q \times ({\mathcal T}(R)\cup\{epsilon\}) \rightarrow 2^Q$ be its transition function. In the evaluation algorithm, we use the product automaton $G \times \mathcal{A}_R$ (in which $\delta': \langle voc(G)\times Q \rangle \times ({\mathcal T}(R)\cup\{epsilon\}) \rightarrow 2^{voc(G)\times Q}$ is its transition function). We construct $G \times \mathcal{A}_R$ as follows:
\begin{itemize}
	\item $\langle u,q \rangle \in voc(G)\times Q$, for every $u \in voc(G)$ and $q \in Q$;
	\item $\langle v,q \rangle \in \delta'(\langle u,p \rangle, s)$ iff $q \in \delta(p,s)$; and one of the following conditions satisfied:
	\begin{itemize*}
		\item $s=$ \smtt{axis} and there exists $a$ s.t. $\langle$\smtt{axis}::$a,v\rangle \in \alpha(u)$
		\item $s=$ \smtt{axis}::$a$ and $\langle$\smtt{axis}::$a,v\rangle \in \alpha(u)$
		\item $s=$ \smtt{axis}::$\psi$ and there exists $b$ s.t. $\langle$\smtt{axis}::$b,v\rangle \in \alpha(u)$ and $b \in \[[\psi\]]_{G}$
	\end{itemize*}
\end{itemize}

Algorithm~\ref{algo:eval} (Eval) solves the evaluation problem for a constrained regular expression $R$ over an RDF graph $G$. 
This algorithm is almost the same as the one in \cite{perez2010a} which solves the evaluation problem for nested regular expressions $R$ over an RDF graph $G$. 
The Eval algorithm calls the Algorithm~\ref{algo:label} (LABEL), which is an adaptation of the LABEL algorithm of \cite{perez2010a} in which we modify only the first two steps. 
These two steps are based on the transformation rules from nSPARQL expressions to cpSPARQL expressions (see \S\ref{sec:cpSPARQLexpr}).

\begin{algorithm}
\caption{{\bf LABEL}(G, exp):} \label{algo:label}
\begin{algorithmic}
\STATE 1. {\bf for each} $axis::[\psi] \in D_0(exp)$ {\bf do}
\STATE 2. 		\hspace{0.5cm} call Label(G, exp') //where $exp' = exp1/self::p$ if $\psi = ?x: <?x, exp1, p>$; $exp' = exp1/self::p$ if $\psi = ?x: <?x, exp1, ?y>$ 
\STATE 3. construct $A_{exp}$, and assume that $q_0$ is its initial state and $F$ is its set of final states
\STATE 4. construct $G \times A_{exp}$
\STATE 5. {\bf for each} state $(u, q_0)$ that is connected to a state $(v, q_f)$ in $G \times A_{exp}$, with $q_f \in F$ {\bf  do}
\STATE 6. 	\hspace{0.5cm} $label(u) := label(u) \cup {exp}$
\end{algorithmic}
\end{algorithm}

The algorithm as the same $O(|G|\times|R|$) time complexity as usual regular expressions \cite{G10,mendelzon1995a} and nested regular expressions \cite{perez2010a} evaluation.

\begin{algorithm}
\caption{{ \bf Eval}$(G,R,\langle a, b \rangle)$} \label{algo:eval}
\begin{algorithmic}
\STATE {\bf Data}: An RDF graph $G$, a constrained regular expression $R$, and a pair $\langle a, b \rangle$.
\STATE {\bf Result}: \textsc{yes} if $\langle a, b \rangle \in \[[R\]]_{G}$; otherwise \textsc{no}.
\vspace{0.4cm}

\STATE {\bf for each} $u \in voc(G)$  {\bf  do}
\STATE \hspace{0.5cm} $label(u) := \emptyset$
\STATE LABEL ($G,R$)
\STATE construct $\mathcal{A}_R$ (assume $q_0:$ initial state and $F:$ set of final states)
\STATE construct the product automaton $G \times \mathcal{A}_R$
\IF{ a state $\langle b,q_f \rangle$ with $q_f \in F$, is reachable from $\langle a,q_0 \rangle$ in $G \times \mathcal{A}_R$} 
\RETURN \textsc{yes};
\ELSE 
\RETURN \textsc{no};
\ENDIF
\end{algorithmic}
\end{algorithm}

\begin{theorem}[Complexity of cpSPARQL regular expression evaluation]\label{th:cpsparqlrecomplexity}
{\emph \textbf{Eval}} solves the evaluation problem for constrained regular expression in time $O(|G|\times|R|$).
\end{theorem}


\subsection{SPARQL queries modulo RDFS with CPSPARQL}\label{sec:cpsparqlrdfs}

Like for nSPARQL, constraints allow for encoding RDF Schemas within queries.

\begin{definition}[RDFS triple pattern expansion \cite{alkhateeb2008b}]\label{def:rdfsexp}
Given an RDF triple $t$, the \emph{RDFS expansion} of $t$, denoted by $\tau(t)$, is defined as:\\
\vspace{-.25cm}
\begin{align*}
\tau(\langle s,\smtt{sc},o\rangle) =& \langle s,\smtt{next::sc}^+,o\rangle\\
\tau(\langle s,\smtt{sp},o\rangle) =& \langle s,\smtt{next::sp}^+,o\rangle\\
\tau(\langle s,\smtt{dom},o\rangle) =& \langle s,\smtt{next::dom},o\rangle\\
\tau(\langle s,\smtt{range},o\rangle) =& \langle s,\smtt{next::range},o\rangle\\
\tau(\langle s,\smtt{type},o\rangle) =& \langle s,\smtt{next::type/next::sc}^* \emph {|} \\
                              & \smtt{edge/}(\smtt{next::sp})^*\smtt{/next::dom/}(\smtt{next::sc})^* \emph{|} \\  
                              &  \smtt{node}^{-1}\emph{/}(\smtt{next::sp})^*\emph{/}\smtt{next::range}\emph{/}(\smtt{next::sc})^*, o\rangle\\
\tau(\langle s,p,o\rangle) =& \langle s,(\smtt{next::}[?x: \{ \langle ?x, (\smtt{next::sp})^*, p \rangle \}]),o\rangle \\ 
         & p\not\in\{\smtt{sp}, \smtt{sc}, \smtt{type}, \smtt{dom}, \smtt{range}\}
\end{align*}
\end{definition}

The RDFS expansion of an RDF triple is a cpSPARQL triple.

The extra variable $?x$ introduced in the last item of the transformation, is only used inside the constraint of the constrained regular expression and so it is not considered to be in $dom(\sigma)$, i.e., only variables occurring as a subject or an object in a CPSPARQL triple pattern are considered in maps (see Definition~\ref{def:triple-evaluation}).
Therefore, the projection operator (SELECT) is not needed to restrict the results of the transformed triple as in the case of PSPARQL \cite{alkhateeb2009a}, as illustrated in the following example.

\begin{example}[SPARQL query transformation]\label{ex:psparql}
Consider the following SPARQL query that searches pairs of nodes connected with a property $p$
\begin{smalltt}
SELECT ?X ?Y
WHERE {?X p ?Y .}
\end{smalltt}
It is possible to answer this query modulo RDFS by transforming this query into the following PSPARQL query: 
\begin{smalltt}
SELECT ?X ?Y
WHERE {?X ?P ?Y . ?P sp* p .}
\end{smalltt}
\end{example}

The evaluation of the above PSPARQL query is the map $\{?X \leftarrow a,$ $?P \leftarrow b,$ $?Y \leftarrow c\}$. So, to actually obtain the desired result, a projection (SELECT) operator must be performed since an extra variable $?P$ is used in the transformation. 
It is argued in \cite{perez2010a} that including the projection (SELECT) operator to the conjunctive fragment of PSPARQL makes the evaluation problem NP-hard. 

On the other hand, the query could be answered by transforming it, with the $\tau$ function of Definition~\ref{def:rdfsexp}, 
 to the following cpSPARQL query (in which there is no need for the projection operator):
\begin{smalltt}
?X next::[?z: { ?z (next::sp)* p }] ?Y
\end{smalltt}

Since the variable $?z$ is used inside the constraint, the answer to this query will be $\{?X \leftarrow a, ?Y \leftarrow b\}$ (see Definition~\ref{def:cresemantics}). 

This has the important consequence that any nSPARQL graph pattern can be translated in a cpSPARQL graph pattern with similar structure and no additional variable. 
Hence, no additional projection operation (SELECT) is required for answering nSPARQL queries in cpSPARQL. 

\begin{theorem}\label{th:taucorrect}
Let $\langle x,p,y \rangle$ be a SPARQL triple pattern with $x, y\in ({\mathcal U} \cup {\mathcal B})$ and $p \in {\mathcal U}$, then $\[[\langle x,p,y \rangle\]]_{G}^{rdfs}$ = $\[[\langle x,\tau(p),y \rangle\]]_{G}$ for any RDF graph $G$.
\end{theorem}

\section{On the respective expressiveness of cpSPARQL and nSPARQL}\label{sec:comparison}\label{sec:compare}

In this section, we compare the expressiveness of cpSPARQL with that of nSPARQL.
We identify several assertions which together show that cpSPARQL is strictly more expressive than nSPARQL and that even if nSPARQL were added projection, it would remain strictly less expressive than CPSPARQL. 
These are the core results of \cite{alkhateeb2014a}.



\subsection{Nested regular expressions (nSPARQL) cannot express all (SPARQL) triple patterns}

Although it is explained in \cite{perez2010a} that SPARQL triple patterns can be encoded by nested regular expressions, triple patterns with three variables (subject, predicate, object) could not be expressed by nested regular expressions since variables are not allowed in nested regular expressions. The reader may wonder whether this is useful or not. The following query is a useful example:

\begin{smalltt}
SELECT * 
WHERE {?s foaf:name "Faisal". ?s ?p ?o .}
\end{smalltt}

That could be used to retrieve all RDF data about a person named "Faisal". However, cpSPARQL triple patterns are proper extension of SPARQL triple patterns and thus the above query could be expressed by the following query: 

\begin{smalltt}
SELECT * 
WHERE {?s next::foaf:name "Faisal". 
       ?s next::]?p:TRUE[ ?o .}
\end{smalltt}

\subsection{nSPARQL without SELECT cannot express all CPSPARQL}

We show in the following that some queries, which can be expressed by CPSPARQL, can only be expressed in nSPARQL with projection (SELECT):

Assume that one wants to retrieve pairs of distinct nodes having a common ancestor. Then the following nSPARQL pattern can express this query:
\vspace{-.25cm}
\begin{align*}
\{ \langle ?person1, &(\smtt{next::ascendant})^+\smtt{/} (\smtt{next}^{-1}\smtt{::ascendant})^+, ?person2 \rangle,\\
& FILTER (!(?person1 = ?person2))\}
\end{align*}
The same query with the restriction that the name of the common ancestor should contain a given family name, for instance "alkhateeb", requires the use of an extra variable to pose the constraint:
\vspace{-.25cm}
\begin{align*}
\{ &\langle ?person1, (\smtt{next::ascendant})^+, ?ancestor  \rangle,\\
   &\langle ?person2, (\smtt{next::ascendant})^+, ?ancestor  \rangle,\\
   & \quad FILTER (!(?person1 = ?person2) \&\& (regex(?ancestor, "\hat{} alkhateeb"))\}
\end{align*}
The evaluation of this graph pattern is the map $\{?person1 \leftarrow p1,$ $?ancestor \leftarrow p3,$ $?person2 \leftarrow p2\}$. Therefore, to obtain the desired result, projection must be performed:
\vspace{-.25cm}
\begin{align*}
\sigma&_{?person1, ?person2}(\\
\{ &\langle ?person1, (\smtt{next::ascendant})^+, ?ancestor  \rangle,\\
   &\langle ?person2, (\smtt{next::ascendant})^+, ?ancestor  \rangle,\\
   & \quad FILTER (!(?person1 = ?person2) \&\& (regex(?ancestor, "\hat{} alkhateeb"))\})
\end{align*}
%

So, the above query cannot be expressed in nSPARQL without the use of SELECT, which is not allowed in nSPARQL \cite{perez2010a}. 
Besides, any SPARQL query that uses SELECT over a set of variables such that there exists at least one existential variable, i.e., a variable not in the SELECT clause, used in a FILTER constraint cannot be expressed by nSPARQL graph patterns.

However, the following CPSPARQL graph pattern could be used to express the above query:
\vspace{-.25cm}
\begin{align*}
\{ &\langle ?person1, (\smtt{next::ascendant})^+\\
& \qquad\smtt{/self::}[?ancestor:  FILTER (regex(?ancestor, "\hat{}alkhateeb" ))]\\
 & \qquad\smtt{/}(\smtt{next}^{-1}\smtt{::ascendant})^+, ?person2  \rangle,  FILTER (!(?person1 = ?person2))\}
 \end{align*}


\subsection{nSPARQL cannot express all cpSPARQL, even with SELECT}

In the following discussion, we show that there exists a cpSPARQL regular expression that cannot be expressed in a nested regular expression as well as some natural and useful queries that can be expressed in CPSPARQL patterns cannot be expressed in nSPARQL patterns even with the SELECT operator.

If one wants to restrict the query of Example~\ref{ex:nsparqlpattern} such that every stop is a city in the same country (for example, France), then the following nested regular expression expresses this query:

$$\langle ?city_1, (next::[(next::sp)^*/self::transport]/self::[next::cityIn/self::France])^+, ?city_2 \rangle$$

This query could also be expressed in the following constrained regular expressions:

$$\langle ?city_1, (next::[\psi_1]/self::[\psi_2])^+, ?city_2 \rangle$$
where $\psi_1= ?x:$ $\{\langle ?x, (next::sp)^*, transport\rangle\}$\\
and $\psi_2= ?x:$ $\{\langle ?x, next::cityIn, France\rangle\}$

If one wants that each stop satisfies a specific constraint, e.g., cities with a population size larger than $20,000$ inhabitants, and each transportation mean belongs to Air France, i.e., its URI is in the airfrance domain name. Then this query is expressed by the following constrained regular expression:

$$P= \langle ?city_1, (next::[\psi_1]/self::[\psi_2])^+, ?city_2 \rangle$$
where $\psi_1= ?x:$ $\{\langle ?x, (next::sp)^*, transport\rangle .$ $FILTER$ $(regex(?x,"www.AirFrance.fr/"))\}$\\
and $\psi_2= ?x:$ $\{\langle ?x, next::size, ?size\rangle . $ $FILTER$ $(?size >20,000)\}$

However, this query cannot be expressed by a nested regular expression,
since it is not possible to apply constraints, such as SPARQL constraints, in the traversed nodes. 
Only navigational constraints can be expressed.


In this case, the variables $?x$ and $?size$  are not exported. 
Hence, the above query can be expressed by a cpSPARQL regular expression without requiring the SELECT operation. 
This cannot be expressed by a nested regular expression.

\begin{theorem}\label{theorem:counterexample}
Not all constrained regular expression $R$ can be expressed as a nested regular expression $R'$ such that 
$\[[R\]]_{G}=\[[R'\]]_{G}$, for every RDF graph $G$.
\end{theorem} 

The type of counter-examples exhibited by the proof of Theorem~\ref{theorem:counterexample} may seem caricatural.
However, it illustrates the capability to apply (non navigational) constraints to values which nSPARQL lacks.
Beside such a minimal example set forth for proving the theorem the same capability is used in more elaborate path queries seen in examples of previous sections (selecting path with intermediate nodes or intermediate predicates satisfying some constraints).

This capability to express constraints on values in path expressions, available in XPath as well, is invaluable for selecting exactly those paths that are useful instead of being constrained to resort to a posteriori selection.
This provides interesting computational properties discussed in Section~\ref{sec:impl}.

The following is another counter-example that could not be expressed as a nested regular expression.

\begin{example}

Consider the following RDF graph representing flights belonging to different airline companies and other transportation means between cities:

\begin{align*}
\{\langle city_1, airfrance:flight_1, city_2 \rangle \\
\{\langle city_2, airfrance:flight_2, city_3 \rangle \\
...\\
\{\langle city_i, anothercomapny:flight_1, city_j \rangle \\
\end{align*}

Assume that one wants to search pairs of cities connected by a sequence of flights belonging to the airfrance company. 
Since there is no way to select (constrain) the transportation means in nested regular expressions, the only way the user can express such a query is to list all flights belonging to airfrance as follows:

\begin{align*}
(airfrance:flight_1|...|airfrance:flight_n)^+\\
\end{align*}

However, this requires the user to know in advance these flights. 
Hence, independent of the RDF graph, the exact meaning of the above query cannot be expressed by nested regular expressions.

\end{example}

\subsection{cpSPARQL can express all nSPARQL}\label{sec:cpSPARQLexpr}

On the other hand, any nested regular expression $R$ could be translated to a constrained regular expression $R_1= trans(R)$ as follows:
\begin{enumerate}
	\item if $R$ is either \smtt{axis} or \smtt{axis}::a, then $trans(R)=R$;
	\item if $R = R_1/R_2$, then $trans(R) = trans(R_1)/trans(R_2)$;
	\item if $R = R_1|R_2$, then $trans(R) = trans(R_1)|trans(R_2)$;
	\item if $R = (R_1)^*$, then $trans(R) = (trans(R_1))^*$;
	\item if $R=exp_1::[exp_2]$, then $trans(R)=exp_1::[\psi]$, such that:	
		\begin{itemize*}
			\item $\psi = ?x:\{\langle ?x, trans(exp_3), p \rangle\}$, if $exp_2=exp_3/self::p$
			\item $\psi = ?x:\{\langle ?x, trans(exp_2), ?y \rangle\}$, otherwise.
		\end{itemize*}
\end{enumerate}


In the last clause of this transformation, when the nested regular expression $R = exp_1::[exp_2]$, it is required to check the existence of two pairs of nodes that satisfies the sub-expression $exp_2$ (see Definition~\ref{def:nresemantics}). Similarly, in cpSPARQL it is necessary to express this nested regular expression as a triple in which the constraint is satisfied by the existence of two pairs of nodes that replace the variables $?x$ and $?y$. 

This transformation process is illustrated by the following example.

\begin{example}[From nSPARQL to cpSPARQL]\label{ex:npsparql}
Consider the following nested regular expression:
$$R_1=(next::[(next::sp)^*/self::transport])^+$$
according to the transformation rules above, the constrained regular expression equivalent to this expression $R_2$
\vspace{-.25cm}
\begin{align*}
=~&trans(R_1)  &\\
=~&trans((next\! :: \![(next\! :: \!sp)^*/self\! :: \!transport])^+)\\
=~&(trans(next\! :: \![(next\! :: \!sp)^*/self\! :: \!transport]))^+\\
=~&next\! :: \![?x\!:\!\{\langle ?x, trans((next\! :: \!sp)^*), transport\rangle\}]\\
=~&next\! :: \![?x\!:\!\{\langle ?x, (trans(next\! :: \!sp))^*, transport\rangle\}]\\
=~&next\! :: \![?x\!:\!\{\langle ?x, (next\! :: \!sp)^*, transport\rangle\}]
\end{align*}
by successively using rules 4, 5, 4, 1, and 5.
\end{example}

\begin{theorem}\label{theorem:cpSPARQLcompletness}
 Any nested regular expression $R$ can be transformed into a constrained regular expression $trans(R)$ such that
 $\[[R\]]_{G}=\[[trans(R)\]]_{G}$, for every RDF graph $G$.
 \end{theorem}

\section{Implementation}\label{sec:impl}

CPSPARQL has been implemented in order to evaluate its feasibility\footnote{The prototype is available at \url{http://exmo.inria.fr/software/psparql}.}. cpSPARQL does not exist as an independent language but is covered by CPSPARQL. This implementation has not been particularly optimised. It passes the W3C compliance tests for SPARQL 1.0 (but 5 tests involving the non implemented \smtt{DESCRIBE} clause).

Experiments have been carried out for evaluating the behaviour of the system and test its ability to correctly answer SPARQL, PSPARQL, and CPSPARQL queries in reasonable time  (against different RDF graph sizes from 5, 10, . . . , up to 100,000 triples in memory graphs). 
In particular, it showed the capability at stake here: answering SPARQL queries with the RDFS semantics.

The implementation has been also tested thoroughly in \cite{Arenas2012a} and the results show that PSPARQL had better performances than other implementations of SPARQL with paths\footnote{The queries and the RDF data that are used for the experimental results can be found in \\ \url{http://www.dcc.uchile.cl/~jperez/papers/www2012/}}. 

It has not been possible to us to compare the performance of our CPSPARQL implementation with other proposals.
Indeed, contrary to CPSPARQL, nSPARQL is not implemented at the moment, so we must leave the experimental comparison for future work. 

However, the experimentation has allowed to make interesting observations.
In particular, the CPSPARQL prototype shows that queries with constraints are answered faster than the same queries without constraints. 
Indeed, CPRDF constraints allow for selecting path expressions with nodes satisfying constraints while matching (on the fly instead of filtering them a posteriori). 
The implemented prototype follows this natural strategy, thus reducing the search space. 
This strategy promises to be always more efficient than a strategy which applies constraints a posteriori.
More details are available in \cite{alkhateeb2008b}.


\section{Related work}\label{sec:relwork}

The closest work to ours, nSPARQL, has been presented and compared in detail in Section~\ref{sec:nsparql} \cite{perez2010a}. However, there are other work which may be considered relevant.

RQL \cite{karvounarakis2002} attempts to combine the relational algebra with some special class hierarchies. It supports a form of transitive expressions over RDFS transitive properties, i.e., \smtt{rdfs:subPropertyOf} and \smtt{rdfs:subClassOf}, for navigating through class and property hierarchies. Versa \cite{versa}, RxPath \cite{rxpath} are all path-based query languages for RDF that are well suited for graph traversal. SPARQLeR \cite{kochut2007a} extends SPARQL by allowing query graph patterns involving path variables. Each path variable is used to capture simple, i.e., acyclic, paths in RDF graphs, and is matched against any arbitrary composition of RDF triples between two given nodes. This extension offers functionalities like testing the length of paths and testing if a given node is in the found paths. SPARQ2L \cite{sparq2L} also allows using path variables in graph patterns. However, these languages 
have not been shown to evaluate queries with respect to RDF Schema and their evaluation procedure has not been proved complete to our knowledge.
Moreover, answering path queries to capture acyclic (simple) paths is NP-complete \cite{mendelzon1995a} (see also \cite{Arenas2012a}).

Path queries (queries with regular expressions) can be translated into recursive Datalog programs over a ternary relation triple $\langle$node, predicate, node$\rangle$, which encodes the graph \cite{abiteboul1997a}. 
This could provide a way to evaluate path queries with Datalog. However, such translations may yield to a Datalog program whose evaluation does not terminate. On the other hand, several techniques can be used to optimize path queries and provide good results in comparison with optimized Datalog programs as shown in \cite{fernandez1998a}.
Recently \cite{cali2009a} extended Datalog in order to cope with querying modulo ontologies. Ontologies are in DL-Lite and, in particular DL-Lite$_{\mathcal{R}}$ which contains the fragment of RDFS considered here. However, this work only considers conjunctive queries which is not sufficient for evaluating SPARQL queries which contains constructs such as UNION, OPT and constraints (FILTER) which are not found in Datalog. \cite{artale2009a} studied from a computational complexity the same fragments with queries containing UNION in addition.
However, given that this fragment is larger than the simple path queries considered in nSPARQL and cpSPARQL,
the complexity is far higher (coNP).

Standardization efforts by the W3C SPARQL working group have defined the notion of inference regime for SPARQL \cite{glimm2010a,glimm2010b}. This notion is relevant to query evaluation modulo RDFS that is exhibited by CPSPARQL and is obviously less relevant to cpSPARQL and nSPARQL. One main difference is that we have departed from the strict definition of ``matching graph patterns'' with the use of path for exploring the graph, and specifically the graph entailed by RDFS. This avoids the use of RDF graph closure on which strict matching is applied. CPSPARQL and nSPARQL use query rewriting for answering queries modulo RDFS, but, unlike DL-Lite rewriting strategies, queries are rewritten by preserving their structure instead of producing unions of conjunctive queries. 

\cite{wudagechekol2011b} studied the static analysis of PSPARQL query containment: determining whether, for any graph, the answers to a query are contained in those of another query. This is achieved by encoding RDF graphs as transition systems and PSPARQL queries as $\mu$-calculus formulas and then reducing the containment problem to testing satisfiability in the logic. 

The language RPL extends nested regular expressions \cite{zauner2010a} to allow boolean node tests. However, using variables in nested regular expressions of nSPARQL requires extending its syntax and semantics. Hence, comparison between variables and values as well as triple patterns with variables in subject, predicate and object are not allowed (see examples in Section~\ref{sec:compare}).


\section{Conclusion}\label{sec:conclusion}

The SPARQL query language has proved to be very successful in offering access to triple stores over SPARQL endpoints all over the web.
It is a critical element of the semantic web infrastructure. However, by limiting it to querying RDF graphs, little consideration has been made of the semantic aspect of RDF. 
In particular, querying RDF graphs modulo RDF Schemas or OWL ontologies is a most needed feature \cite{glimm2010a}.

One possible approach for querying an RDFS graph $G$ in a sound and complete way is by computing the closure graph of $G$, i.e., the graph obtained by saturating $G$ with all informations that can be deduced using a set of predefined rules called RDFS rules, then evaluating the query $Q$ over the closure graph. However, this approach takes time proportional to $|Q| \times |G|^2$ in the worst case \cite{munoz2009a}.

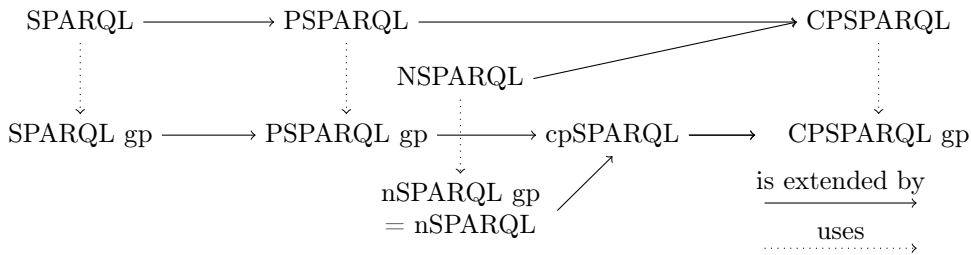
\begin{figure*}[!htp]
 \begin{center}
\begin{tikzpicture}

\draw (0,1.5) node (sparql) {SPARQL};
\draw (3.5,1.5) node (psparql) {PSPARQL};
\draw (10.5,1.5) node (cpsparql) {CPSPARQL};
\draw (5,.75) node (nsparql) {NSPARQL};

\draw[->] (sparql.east) -- (psparql.west);
\draw[->] (psparql.east) -- (cpsparql.west);
\draw[->] (nsparql.east) -- (cpsparql.west);

\draw (0,0) node (sgp) {SPARQL gp};
\draw (3.5,0) node (pgp) {PSPARQL gp};
\draw (10.5,0) node[text width=3cm,text badly centered] (cpgp) {CPSPARQL gp};
\draw (7,0) node (cpsparqlgp) {cpSPARQL};
\draw (5,-1) node[text width=2.3cm,text badly centered] (ngp) {nSPARQL gp = nSPARQL};

\draw[->] (sgp.east) -- (pgp.west);
\draw[->] (pgp.east) -- (cpsparqlgp.west);
\draw[->] (cpsparqlgp.east) -- (cpgp.west);
\draw[->] (ngp.east) -- (cpsparqlgp.south);
\draw[->] (cpsparqlgp.east) -- (cpgp.west);

\draw[dotted,->] (sparql.south) -- (sgp.north);
\draw[dotted,->] (psparql.south) -- (pgp.north);
\draw[dotted,->] (cpsparql.south) -- (cpgp.north);
\draw[dotted,->] (nsparql.south) -- (ngp.north);

\draw[->] (9,-.9) -- node[above] {is extended by} (11,-.9);
\draw[dotted,->] (9,-1.5) -- node[above] {uses} (11,-1.5);

\end{tikzpicture}						
\end{center}
\caption{Query languages and their graph patterns.}\label{fig:qlgp}
\end{figure*}

Over the years, several other languages, i.e., PSPARQL \cite{alkhateeb2009a}, nSPARQL \cite{perez2010a} and CPSPARQL \cite{alkhateeb2008b}, have been shown able to deal with RDFS graphs without computing the closure.
They all use different variations of path regular expressions as triple predicates and adopt a semantics based on checking the existence of paths (without counting them) in the RDF graph. 
In order to ease their comparison, we defined cpSPARQL as very close to nSPARQL and NSPARQL as very close to CPSPARQL.

Figure~\ref{fig:qlgp} shows the position of the various languages. 
nSPARQL and cpSPARQL are good navigational languages for RDF(S). 
However, cpSPARQL is an extension of SPARQL graph patterns, while nSPARQL does not contain all SPARQL graph patterns. Moreover, using such a path language within the SPARQL structure allows for properly extending SPARQL. 
Some features (such as filtering nodes inside expressions) are very simple to add to the syntax and semantics of nested regular expressions. 

More precisely, we showed that cpSPARQL, the fragment of CPSPARQL which is sufficient for capturing RDFS semantics, admits an efficient evaluation algorithm while the whole CPSPARQL language is in theory as efficient as SPARQL is. 
Moreover, we compared cpSPARQL with nSPARQL and showed that cpSPARQL is strictly more expressive than nSPARQL. 

It is likely that more expressive fragments of CPSPARQL graph patterns keeping the same complexity may be found.
In particular, we did not kept the capability to express the constraints existentially or universally.
This may be useful, for instance, to filter families all children of which are over 18 or families one children of which is over 18.

This work can also be extended in the schema or ontology language on which it is applied.
OWL 2 opens the door to many OWL fragments for which it should be possible to design query evaluation procedures. 

\clearpage

\appendix\label{sec:appendix}

\section{Proofs}\label{app:proofs}





%

\subsection{Induction lemma}

All completeness proofs below follow the same style: because most of the results that we use are based on entailment of graphs, i.e., basic graph patterns, we need to promote these results to all graph patterns. This is done easily by defining query-structure preserving transformations and using Lemma~\ref{lem:induction}.

\begin{definition}[Query-structure preserving transformation]
A transformation $\psi$ on RDFS graphs and SPARQL graph patterns is said to be \emph{query structure preserving} if and only if:
\begin{align}
\psi(G\models P) &= \psi(G)\models \psi(P)\\
\psi(P\texttt{ AND }P') &= \psi(P)\texttt{ AND }\psi(P')\\
\psi(P\texttt{ UNION }P') &= \psi(P)\texttt{ UNION }\psi(P')\\
\psi(P\texttt{ OPT }P') &= \psi(P)\texttt{ OPT }\psi(P')\\
\psi(P\texttt{ FILTER }K) &= \psi(P)\texttt{ FILTER }K
\end{align}
\end{definition}

As its name indicates, such a transformation preserves the structure of graph patterns. This is the case of most transformations proposed here since they are defined on basic graph patterns and extended to queries by applying them to basic graph patterns. The structure of $P$ must be preserved, but not be isomorphic to that of $\psi(P)$. For instance, the transformation $\tau$ may introduce extra \texttt{UNION} in the resulting pattern.

The induction lemma shows that if an entailment relation is complete for basic graph patterns, then it is complete for all graph patterns.

\begin{lemma}[Induction lemma]\label{lem:induction}
Let $\psi$ be a query-structure preserving transformation,
if for all RDFS graph $G$, basic graph pattern $B$ and map $\sigma$, $\psi(G)\models\sigma(\psi(B)) \text{ iff } G\models \sigma(B)$,
then for all graph pattern $P$, $\psi(G)\models\sigma(\psi(P)) \text{ iff } G\models \sigma(P)$.
\end{lemma}

\begin{proof}[Proof of Lemma~\ref{lem:induction}]
The lemma itself is proved by induction:

\textbf{Base step} For $B$ a basic graph pattern and $G$ an RDFS graph, $\psi(G)\models\sigma(\psi(B)) \text{ iff } G\models \sigma(B)$, by hypothesis,

\textbf{Induction step} If for $P$ and $P'$, graph patterns, and $G$ RDFS graph and $\sigma$ map, $\psi(G)\models\sigma(\psi(P)) \text{ iff } G\models \sigma(P)$ and  $\psi(G)\models\sigma(\psi(P')) \text{ iff } G\models \sigma(P')$, then:
\begin{align*}
G\models \sigma(P~&\texttt{AND}~P') \text{ iff } G\models \sigma(P)\text{ and }G\models \sigma(P')\\
	 &\text{ iff } \psi(G)\models \psi(\sigma(P))\text{ and }\psi(G)\models \psi(\sigma(P'))\\
	 &\text{ iff } \psi(G)\models \sigma(\psi(P~\texttt{AND}~P'))\\
G\models \sigma(P~&\texttt{UNION}~P') \text{ iff } G\models \sigma(P)\text{ or }G\models \sigma(P')\\
	 &\text{ iff } \psi(G)\models \sigma(\psi(P))\text{ or }G\models \sigma(\psi(P'))\\
	&\text{ iff } \psi(G)\models \sigma(\psi(P~\texttt{UNION}~P'))\\
\begin{split}
G\models \sigma(P~&\texttt{OPT}~P') \text{ iff } G\models \sigma(P)\text{ and } [G\models \sigma(P')\\ 
& \text{ or } \forall \sigma'; G\models \sigma'(P'), \sigma\bot \sigma']
\end{split}\\
\begin{split}
	 &\text{ iff } \psi(G)\models \sigma(\psi(P))\text{ and } [\psi(G)\models \sigma(\psi(P'))\\ 
& \text{ or } \forall \sigma'; \psi(G)\models \sigma'(\psi(P')), \sigma\bot \sigma']
\end{split}\\
 &\text{ iff } \psi(G)\models \sigma(\psi(P~\texttt{OPT}~P'))\\
G\models \sigma(P~&\texttt{FILTER}~K) \text{ iff } G\models \sigma(P)\text{ and }\sigma(K)=\top\\
	&\text{ iff } \psi(G)\models \sigma(\psi(P))\text{ and }\sigma(K)=\top\\
	&\text{ iff } \psi(G)\models \sigma(\psi(P~\texttt{FILTER}~K)) 
\end{align*}
The only difficult point is in the \texttt{OPT} part, but since the induction step strictly preserves the set of maps satisfying a graph pattern, the universal quantification holds.
\end{proof}

\subsection{Completeness of partial non reflexive RDFS closure}

We first have to extend the completeness proof of the partial closure (Proposition~\ref{thm:rdfsentailment}) to the non reflexive case.

\begin{proof}[Proof of Proposition~\ref{thm:rdfnrxsentailment}]
The proof can be derived from Proposition~\ref{thm:rdfsentailment}.
Indeed, it should be shown that suppressing rules [RDFS8a] and [RDFS12a] does suppress all and only consequences
of the reflexivity of \smtt{sc} and \smtt{sp}.

[RDFS12a] generates $\langle c~\smtt{sc}~c\rangle$ which can be further used by [RDFS12b] and [RDFS11]. However, these two rules would only generate triples that are already in their premises. Hence the only new generated triple is the reflexivity triple.

[RDFS8a] generates $\langle p~\smtt{sp}~p\rangle$ which can be further used by [RDFS8b], [RDF2] and [RDFS9]. Similarly, 
these three rules would only generate triples that are already in their premises or in axiomatic triples. Hence the only new generated triple is the reflexivity triple.

Finally, both triples could be consumed by rules [RDFS6] and [RDFS7]. However, these rules would require constraining the \smtt{sp} or \smtt{sc} relations through \smtt{dom} or \smtt{range} statements respectively. This is not possible in genuine RDFS graphs.

In the other direction, the only way a model can contain $\langle p, p\rangle\in I_{EXT}(\iota'(\smtt{sp}))$ or $\langle c, c\rangle\in I_{EXT}(\iota'(\smtt{sc}))$ is through Definition~\ref{def:rdfsmodel} constraints:
\begin{itemize*} 
\item (1a) it is a triple of $G$. Then it is still in $\hat{G}\backslash\backslash H$;
\item (4a) and (5a) by reflexivity;
\item (4b), (6b) and (6c) which only apply to generate reflexive \smtt{sc} or \smtt{sp} statements when constraints are added to \smtt{sc} or \smtt{sp}.
\end{itemize*}
\end{proof}


\begin{proof}[Proof of Corrolary~\ref{thm:answersparqlmodulo}]
The proof is a simple consequence of Proposition~\ref{thm:rdfnrxsentailment}. Proving that:
$$\mathcal{A}^{\#}(\vec{B}, G, P)=\mathcal{A}(\vec{B}, \hat{G}\backslash\backslash P, P)$$
by Definition~\ref{def:answersparqlmodulo} and Proposition~\ref{prop:sparql}, is equivalent to proving that:
$$\{\sigma|_{\vec{B}}^{\vec{B}}| G \models_{RDFS}^{\text{nrx}} \sigma(P) \}=\{\sigma|_{\vec{B}}^{\vec{B}}| \hat{G}\backslash\backslash P \models_{RDF} \sigma(P) \}$$
and Proposition~\ref{thm:rdfnrxsentailment} together with Lemma~\ref{lem:induction} means that this is true if $G$ is satisfiable and genuine.
\end{proof}

\subsection{Completeness and complexity of the PSPARQL RDFS query encoding}

\begin{proof}[Proof of Proposition~\ref{thm:cpsparqlrdfsans}]
We use the same pattern as before, using Lemma~\ref{lem:psparqlbasecompleteness} to be proved below. 
Proving that:
$$\mathcal{A}^{\#}(\vec{B}, G, P)=\mathcal{A}^{\star}(\vec{B}, G, \tau(P))$$
by Definition~\ref{def:answersparqlmodulo} and Definition~\ref{def:answerpsparql}, is equivalent to proving that:
$$\{\sigma|_{\vec{B}}^{\vec{B}}| G \models_{RDFS}^{\text{nrx}} \sigma(P) \}=\{\sigma|_{\vec{B}}^{\vec{B}}| G \models_{PSPARQL} \sigma(\tau(P)) \}$$
and Lemma~\ref{lem:psparqlbasecompleteness} together with Lemma~\ref{lem:induction} means that this is true.
\end{proof}

\input{completness-psparql}

\begin{proof}[Proof of Proposition~\ref{prop:PSPARQLcmpl}]
The complexity of answering queries modulo RDF Schema through the PSPARQL RDFS encoding is the sequential combination of that of the two components:
encoding ($\tau$) and evaluating the resulting query. The complexity of the encoding is linear in terms of triples in the query graph pattern $P$, moreover its size is not more than seven times that of $P$ (if all triples are type statements).
Since the complexity of PSPARQL query answering has been proved to be PSPACE-complete in \cite{alkhateeb2009a}, the complexity of \textsc{$\mathcal{A}^{\star}$-Answer checking} is PSPACE-complete.
\end{proof}

\subsection{Completeness of NSPARQL query answering}

\begin{proof}[Proof of Proposition~\ref{prop:nsparqlanswer}]
This is similar to the previous proof.
The proof is a simple consequence of Proposition~\ref{prop:nsparqlcompl}. 
Proving that:
$$\mathcal{A}^{\#}(\vec{B}, G, P)=\mathcal{A}^{o}(\vec{B}, G, \phi(P))$$
by Definition~\ref{def:answersparqlmodulo} and Definition~\ref{def:nsparqlanswer}, is equivalent to proving that:
$$\{\sigma|_{\vec{B}}^{\vec{B}}| G \models_{RDFS}^{\text{nrx}} \sigma(P) \}=\{\sigma|_{\vec{B}}^{\vec{B}}| G \models_{nSPARQL} \sigma(\phi(P)) \}$$
and Proposition~\ref{prop:nsparqlcompl} together with Lemma~\ref{lem:induction} means that this is true.
\end{proof}





\subsection{Complexity of cpSPARQL evaluation}


\begin{proof}[Proof of Theorem~\ref{th:cpsparqlrecomplexity}]
Let $R$ be the a constrained regular expression, $G$ be an RDF graph,  $\langle a, b \rangle$ be a pair of nodes, and $A_R$ be the automaton recognizing the language of $R$. 

The automaton of $R$ can be constructed as described in \S\ref{sec:complexityresult} in \textsc{nlogspace} (as for the usual automata \cite{G10,mendelzon1995a}). 
For simplicity and without loss of generality, we use the \smtt{next} axis to illustrate the construction of the product automaton. 
This is because the axis determines the node to be checked (subject, predicate or object) and thus does not affect the construction. 
The construction of the product automaton is done as follows:  
\begin{itemize}
\item If $R = axis::]?x: TRUE[$ then checking whether the pair $\langle a,b \rangle$ is in $\[[R\]]_G$ can be done in $O(|G|)$ since it is sufficient to substitute each node $n$ to $?x$ and check whether $\langle a,n,b \rangle$ is in $G$ (according to the axis).

\item Otherwise, call the Eval algorithm (where $D_0$ is defined as done in \cite{perez2010a}). If $\langle s_i, next::[FILTER (?x)], s_j \rangle \in A_R$ and $\langle n_i, next::p, n_j \rangle \in G$, then add $\langle s_j, n_j \rangle$ to the product automaton if $p$ satisfies the SPARQL filter constraint by substituting only the node $p$ to the variable $?x$. Checking if a node $n$ satisfies a SPARQL filter constraint can be done in $O(1)$. 
 
Additionally, if $\langle ?x, next\!::\!p, ?y \rangle. FILTER (?x,?y)$ $\in A_R$ and $\langle n_i, next::p, n_j \rangle \in G$, then add $\langle s_j, n_j \rangle$ to the product automaton if $\langle n_i, next::p, n_j \rangle \in G$ and the SPARQL filter constraint is satisfied by substituting the node $n_i$ to the variable $?x$ and $n_j$ to the variable $?y$.
	
	
\end{itemize}

So, the product automaton ($G \times \mathcal{A}_R$) can be obtained in time $O(|G|\times|R|)$. 
Hence, checking if the pair $\langle a, b \rangle \in \[[R\]]_{G}$ is equivalent to checking if the language accepted by ($G \times \mathcal{A}_R$) is not empty, which can be done in $O(|G|\times|R|)$ (as in the case of usual regular expressions \cite{G10,mendelzon1995a}).
\end{proof}

\subsection{Correctness and completeness of RDFS translation to PSPARQL}


The proof of Theorem~\ref{th:taucorrect} follows from the results in \cite{perez2010a} except the last step.

\begin{proof}[Proof of Theorem~\ref{th:taucorrect}]
We need to prove only the last step since all other transformation steps are the same as the ones in \cite{perez2010a}. 
That is $\langle \sigma(x),\sigma(y)\rangle \in$ $\[[\langle x,p,y \rangle\]]_{G}^{rdfs}$ iff $\langle \sigma(x),\sigma(y)\rangle \in$ $\[[\langle x,\tau(p),y \rangle\]]_{G}$.
 
\begin{itemize}
	\item ($\Rightarrow$) Assume that $\langle \sigma(x),\sigma(y) \rangle \in$ $\[[\langle x,p,y \rangle\]]_{G}^{rdfs}$. In this case, there exists $p_1$ such that ($p_1$ $sp$ $p_2$ $sp$ $\ldots$ $sp$ $p_n=p$) and $\langle \sigma(x), p_1,\sigma(y) \rangle \in G$ as well as $\langle \sigma(x),$ \smtt{next::}$p_1,\sigma(y) \rangle \in G$. Let us consider now the transformed triple $\tau(t) =\langle x,($\smtt{next::}$\psi),y\rangle$ (where $\psi=[?p: \{ \langle ?p,$ (\smtt{next::sp}$)^*, p \rangle \}]$). The maps for the variable $?p$ will be $\{ \langle ?p,p_i \rangle \ | \ i=1, \ldots, n\}$ (since $\[[\psi\]]_{G}= \{ \langle p_i,p \rangle \ | \ i=1, \ldots, n\}$). According to Definitions~\ref{def:cresemantics} and \ref{def:triple-evaluation}, 
	 $\langle \sigma(x),\sigma(y)\rangle \in$ $\[[\langle x,($\smtt{next::}$\psi),y\rangle\]]_{G}$ iff $\langle \sigma(x),\sigma(y)\rangle \in G$ and $p_1 \in \[[\psi\]]_{G}$, and this condition holds.  
	
	\item ($\Leftarrow$) We have to prove that, if $\langle \sigma(x),\sigma(y)\rangle \! \in \! \[[\langle x,$ $($\smtt{next::}[$\psi]),y\rangle\]]_{G}$ (with $\psi=?p: \{ \langle ?p,$ (\smtt{next\!::\!sp}$)^*, p \rangle \}$), then $\langle \sigma(x),\sigma(y)\rangle \!\in\!\[[\langle x,p,y$ \\ $ \rangle\]]_{G}^{rdfs}$. Assume that $\langle \sigma(x),\sigma(y)\rangle \in$ $\[[\langle x,($\smtt{next::}$\psi), y\rangle\]]_{G}$. In this case,  there exists $p_1$ such that $\langle \sigma(x),$ \smtt{next::}$p_1,\sigma(y)\rangle \in G$ and $p_1 \in \[[\psi\]]_{G}$, that is, $\langle p_1,$\smtt{next::sp}$,p_2 \rangle$, $\ldots$, $\langle p_{n-1},$\smtt{next::sp}$,p_n=p \rangle \in G$. Therefore, $\langle \sigma(x),\sigma(y) \rangle \in$ $\[[\langle x,p,y \rangle\]]_{G}^{rdfs}$ since $\langle p_1,($\smtt{next::sp}$)^*,p \rangle$ and $\langle \sigma(x),$ \smtt{next::}$p_1,$ $\sigma(y)\rangle \in G$.      		
\end{itemize}
\end{proof}

\subsection{Expressiveness of nSPARQL and cpSPARQL}

\begin{proof}[Proof of Theorem~\ref{theorem:counterexample}]
Consider, without loss of generality, RDF graphs containing a predicate $s$ whose range is the set of integers.
If one wants to select nodes which have a $s$-transition whose value is over $3$, this could be expressed by the following constrained regular expression:
$$R \!=\! self\!::\![?s\!:\!\{\langle ?n, next\!::\!s, ?s\rangle . FILTER (?s\! >3)\}]$$
Consider a graph $G$ with two triples $\langle u, s, 2\rangle$ and $\langle v, s, 4\rangle$.
The evaluation of $R$ will return $\[[R\]]_{G} = \{\langle v, v\rangle\}$.

A nSPARQL nested regular expression $R'$ corresponding to $R$, should be able to select the pair $\langle v, v\rangle$ as an answer.
However, the two subgraphs made of the triples in $G$ are isomorphic with respect to their structure.
Hence, any nSPARQL nested regular expression retrieving one of them (a node which is the source of a $s$-edge) will retrieve both of them.

Even assuming that literals are followed and may be constrained by value, which is not the case in the current definition of nSPARQL, it would be necessary to enumerate the $s$-values larger than $3$ (say $4, 5\dots$) to design an expression such as:
$$R'=self::[next::s/self::(4|5|\ldots)]$$
However, there is an infinite number of such values and for the queries to be strictly equivalent, i.e., to provide the same answers for any graph, it is necessary to cover them all.
Indeed, if one value is missing, then it is possible to create a graph $G$ for which the answers to $R$ and $R'$ do not coincide.

It is thus not possible to express a query equivalent to $R$ in nSPARQL.
\end{proof}

\begin{proof}[Proof of Theorem~\ref{theorem:cpSPARQLcompletness}]
The equivalence of the cpSPARQL encoding of nested regular expressions (trans) is given by induction on the structure of nested regular expressions. 
\begin{itemize}
	\item if $R$ is either \smtt{axis} or \smtt{axis}::a, then $trans(R)=R$ and thus $\[[R\]]_G = \[[trans(R)\]]_G$.
	
	\item Now assume that $\[[R1\]]_G = \[[trans(R1)\]]_G$ and $\[[R2\]]_G = \[[trans(R2)\]]_G$,
 then $\[[R1|R2\]]_G = $ $\[[trans(R1)\]]_G$$ \cup $ $\[[trans(R2)\]]_G = $ $\[[trans(R1)|$ $trans(R2)\]]_G$ $ = $ $\[[trans(R1|R2)\]]_G$ (based the definition of regular languages). The same applies for the concatenation $\[[R1/R2\]]_G$ and the closure $(R_1)^*$.
 
 \item If $R=R_1::[R_2]$, then $trans(R)=R_1::[\psi]$, where $\psi = ?x:\{\langle ?x, trans(R_2), ?y \rangle\}$. Based on Definition~\ref{def:nresemantics},  $\[[R_1::[R_2]\]]_{G}$ = $~\{\langle x, y\rangle$ $|  \exists z,w$ $\wedge \langle z,w\rangle\in \[[R_2\]]_{G}\}$. If $\langle z,w\rangle\in \[[R_2\]]_{G}$, then $\langle z,w\rangle\in \[[trans(R2)\]]_{G}$ by substituting $z$ and $w$ to the variables $?x$ and $?y$, respectively (Definitions~\ref{def:satconstarint} and \ref{def:cresemantics}). 

\end{itemize}
\end{proof}

%

\bibliographystyle{abbrv}
\bibliography{semweb,cpsparql} 

\clearpage
\setcounter{tocdepth}{2}

\tableofcontents

\end{document}

%% file: completness-psparql.tex
The proof is longer than the previous ones because it does not rely on externally proved propositions. Instead, we need to prove the following lemma

\begin{lemma}\label{lem:psparqlbasecompleteness}
Given a basic SPARQL graph pattern $P$ and an RDFS graph $G$,
$$ G\models_{RDFS}^{\text{nrx}} P \text{ iff } G\models_{PSPARQL} \tau(P)$$
\end{lemma}

The proof relies on the notion of PRDF homomorphism, a particular sort of map \cite{alkhateeb2009a}:

\begin{definition}[PRDF homomorphism] \label{def:prdfhomomorphism}
Let $G$ be a GRDF graph, and $H$ be a PRDF graph. A PRDF homomorphism from $H$ into $G$ is a map $\pi$ from $term(H)$ into $term(G)$ such that $\forall \langle s,R,o \rangle \in H$, either
\begin{enumerate}[(i)]
\item [($i$)] the empty word $\epsilon \in L^*(R)$ and $\pi(s)=\pi(o)$; or
\item [($ii$)] $\exists \langle n_0,p_1,n_1 \rangle ,\ldots, \langle n_{k-1},p_k,n_k\rangle$ in $G$ such that $n_0=\pi(s)$, $n_k=\pi(o)$, and $p_1 \cdot \ldots \cdot p_k \in L^*(\pi(R))$.
\end{enumerate} 
\end{definition}

\begin{proof}[Proof of Lemma~\ref{lem:psparqlbasecompleteness}]
Since the answers to a graph pattern can be made by joining the answers to its triple patterns \cite{alkhateeb2009a}, it is sufficient to show that answering a triple $t$ is equivalent to answering its transformed triple $t'=\tau(t)$. Hence we consider each case of $\tau$:

\begin{itemize}
	\item Let $t = \langle \mathcal{X}, sc, \mathcal{Y} \rangle$, then $t' = \langle \mathcal{X}, sc^+, \mathcal{Y} \rangle$
	
	($\Rightarrow$) If $\sigma \in \mathcal{A}^{\#}(\vec{B}, G, t)$, then either $\langle \sigma(\mathcal{X}), sc, \sigma(\mathcal{Y}) \rangle \in G$ or $\langle \sigma(\mathcal{X})=\mathcal{X}_0, sc, \mathcal{Y}_0 \rangle$, $\ldots$, $\langle \mathcal{X}_n, sc, \sigma(\mathcal{Y})=\mathcal{Y}_n \rangle \in G$.
 In both cases, $\sigma \in \mathcal{A}^{\star}(\vec{B}, G, t')$.

($\Leftarrow$) If $\sigma \in \mathcal{A}^{\star}(\vec{B}, G, t')$, then $\langle \sigma(\mathcal{X})=\mathcal{X}_0, sc, \mathcal{Y}_0 \rangle$, $\ldots$, $\langle \mathcal{X}_n, sc, \sigma(\mathcal{Y})=\mathcal{Y}_n \rangle \in G$. In what follows, $\langle \sigma(\mathcal{X}), sc, \sigma(\mathcal{Y}) \rangle \in \hat{G}$. So, $\sigma \in \mathcal{A}^{\#}(\vec{B}, G, t)$.

\item Let $t = \langle \mathcal{X}, sp, \mathcal{Y} \rangle$, then $t' = \langle \mathcal{X}, sp^+, \mathcal{Y} \rangle$
	
	($\Rightarrow$) If $\sigma \in \mathcal{A}^{\#}(\vec{B}, G, t)$, then either $\langle \sigma(\mathcal{X}), sp, \sigma(\mathcal{Y}) \rangle \in G$ or $\langle \sigma(\mathcal{X})=\mathcal{X}_0, sp, \mathcal{Y}_0 \rangle$, $\ldots$, $\langle \mathcal{X}_n, sp, \sigma(\mathcal{Y})=\mathcal{Y}_n \rangle \in G$.
 In both cases, $\sigma \in \mathcal{A}^{\star}(\vec{B}, G, t')$.

($\Leftarrow$) If $\sigma \in \mathcal{A}^{\star}(\vec{B}, G, t')$, then $\langle \sigma(\mathcal{X})=\mathcal{X}_0, sp, \mathcal{Y}_0 \rangle$, $\ldots$, $\langle \mathcal{X}_n, sp, \sigma(\mathcal{Y})=\mathcal{Y}_n \rangle \in G$. In what follows, $\langle \sigma(\mathcal{X}), sp, \sigma(\mathcal{Y}) \rangle \in \hat{G}$. So, $\sigma \in \mathcal{A}^{\#}(\vec{B}, G, t)$.

\item Let $t = \langle \mathcal{X}, p, \mathcal{Y} \rangle$, then $t'= \{\langle \mathcal{X},?p,\mathcal{Y} \rangle, \langle ?p,\smtt{sp}^*,p\rangle\}$.

	($\Rightarrow$) If $\sigma \in \mathcal{A}^{\#}(\vec{B}, G, t)$, then either $\langle \sigma(\mathcal{X}), p, \sigma(\mathcal{Y}) \rangle \in G$ or $\langle \sigma_1(\mathcal{X}), \sigma_1(?p)=p_0, \sigma_1(\mathcal{Y}) \rangle$, $\langle p_0, sp, p_1 \rangle$, $\ldots$, $\langle p_{n-1}, sp, p_n=p \rangle \in G$. In the first case, the map $\sigma$ is a PRDF homomorphism from $t'$ into $G$. In the second case, the map $\sigma_1$ is a PRDF homomorphism from $t'$ into $G$. The restriction of $\sigma_1$ to the variables of the query $\vec{B}$ is $\sigma$. Hence in both cases, $\sigma \in \mathcal{A}^{\star}(\vec{B}, G, t')$.

($\Leftarrow$) If $\sigma|_{\vec{B}}^{\vec{B}} \in \mathcal{A}^{\star}(\vec{B}, G, t')$, then $\langle \sigma|_{\vec{B}}^{\vec{B}}(\mathcal{X}), \sigma|_{\vec{B}}^{\vec{B}}(?p)=p_0, \sigma|_{\vec{B}}^{\vec{B}}(\mathcal{Y}) \rangle$, $\langle p_0, sp, p_1 \rangle$, $\ldots$, $\langle p_{n-1}, sp, p_n=p \rangle \in G$. In what follows, $\langle \sigma|_{\vec{B}}^{\vec{B}}(\mathcal{X}), p, \sigma|_{\vec{B}}^{\vec{B}}\mathcal{Y}) \rangle \in \hat{G}$. So, $\sigma|_{\vec{B}}^{\vec{B}} \in \mathcal{A}^{\#}(\vec{B}, G, t)$.

\item Let $t = \langle \mathcal{X}, type, \mathcal{Y} \rangle$, 

($\Rightarrow$) If $\sigma \in \mathcal{A}^{\#}(\vec{B}, G, t)$, then at least one of the following cases is satisfied:
\begin{itemize}
	\item The triples $\langle \sigma(\mathcal{X}), type, o_1 \rangle$, $\langle o_1, sc, o_2 \rangle$, $\ldots$, $\langle o_{n-1}, sc, o_n = \sigma(\mathcal{Y}) \rangle$ belong to $G$. So that $\sigma$ is a PRDF homomorphism from the first part of $t'$ into $G$ and $\sigma \in \mathcal{A}^{\star}(\vec{B}, G, t')$.
	
	\item The triples $\langle \sigma_1(\mathcal{X}), \sigma_1(?p_1)=p_1, \sigma_1(?y)=y \rangle$, $\langle p_1, sp, p_2 \rangle$, $\ldots$, $\langle p_{n-1}, sp, p_n \rangle$, $\langle p_n, dom, o_1 \rangle$, $\langle o_1, sc, o_2 \rangle$, $\ldots$, $\langle o_{n-1}, sc, o_n = \sigma_1(\mathcal{Y}) \rangle$ belong to $G$. So that $\sigma_1$ is a PRDF homomorphism from the second part of $t'$ into $G$. The restriction of $\sigma_1$ to the variables of the query $\vec{B}$ is $\sigma$. Hence $\sigma \in  \mathcal{A}^{\star}(\vec{B}, G, t')$.

	\item The triples $\langle \sigma_1(?y)=y , \sigma_1(?p_1)=p_1, \sigma_1(\mathcal{X}) \rangle$, $\langle p_1, sp, p_2 \rangle$, $\ldots$, $\langle p_{n-1}, sp, p_n \rangle$, $\langle p_n, range, o_1 \rangle$, $\langle o_1, sc, o_2 \rangle$, $\ldots$, $\langle o_{n-1}, sc, o_n = \sigma_1(\mathcal{Y}) \rangle$ belong to $G$. So that $\sigma_1$ is a PRDF homomorphism from the second part of $t'$ into $G$. The restriction of $\sigma_1$ to the variables of the query $\vec{B}$ is $\sigma$. Hence $\sigma \in  \mathcal{A}^{\star}(\vec{B}, G, t')$.
	
\end{itemize}

($\Leftarrow$) If $\sigma|_{\vec{B}}^{\vec{B}} \in \mathcal{A}^{\star}(\vec{B}, G, t')$, then there exist three cases. The first case is that the triples
$\langle \sigma|_{\vec{B}}^{\vec{B}}(\mathcal{X}), type, o_1 \rangle$, $\langle o_1, sc, o_2 \rangle$, $\ldots$, $\langle o_{n-1}, sc, o_n = \sigma|_{\vec{B}}^{\vec{B}}(\mathcal{Y}) \rangle$ belong to $G$. Using the sub-class RDFS rules, $\langle \sigma(\mathcal{X}), type, \sigma(\mathcal{Y}) \rangle \in \hat{G}$. In the second case, the triples $\langle \sigma(\mathcal{X}), \sigma(?p_1)=p_1, y \rangle$, $\langle p_1, sp, p_2 \rangle$, $\ldots$, $\langle p_{n-1}, sp, p_n \rangle$, $\langle p_n, dom, o_1 \rangle$, $\langle o_1, sc, o_2 \rangle$, $\ldots$, $\langle o_{n-1}, sc, o_n = \sigma(\mathcal{Y}) \rangle$ belong to $G$. Using the sub-property RDFS rules, the triples $\langle \sigma(\mathcal{X}), p_2, y \rangle$, $\ldots$, $\langle \sigma(\mathcal{X}), p_n, y \rangle \in \hat{G}$. So that $\langle \sigma(\mathcal{X}), type, o_1 \rangle$, $\ldots$, $\langle \sigma(\mathcal{X}), type, o_n=\sigma(\mathcal{Y}) \rangle \in \hat{G}$ (it is easy to prove the same thing for the third case with \smtt{range} instead of \smtt{dom}). Hence $\sigma|_{\vec{B}}^{\vec{B}} \in \mathcal{A}^{\#}(\vec{B}, G, t)$.

\end{itemize}

\end{proof}